\documentclass[12pt,draftclsnofoot,onecolumn]{IEEEtran}

\usepackage{cite}
\usepackage[hyperfootnotes=false,hyperfigures=false,hidelinks]{hyperref}
\usepackage{stmaryrd}
\usepackage{footnote}
\usepackage{graphicx}
\usepackage{tikz,pgfplots}
\usetikzlibrary{positioning,shapes.misc,shapes.arrows}
\usetikzlibrary{arrows}
\usetikzlibrary{calc}
\pgfplotsset{compat=1.8}
\usepackage{mathtools}
\usepackage{amsmath,amsfonts,amssymb,amsthm}
\usepackage[utf8]{inputenc}
\usepackage[english]{babel}
\usepackage{color}
\usepackage[T1]{fontenc}
\usepackage{subcaption}

\newtheorem{definition}{Definition}
\newtheorem{theorem}{Theorem}

\newtheorem{lemma}[theorem]{Lemma}
\newtheorem{corollary}{Corollary}

\newtheorem{remark}{Remark}

\newcommand{\defeq}{\triangleq}

\def\Pr{\mathrm{P}}
\def\P{\phi}

\newcommand{\HM}{\pmb{H}}
\newcommand{\IM}{\pmb{I}}
\newcommand{\Id}{\pmb{{I}}}
\newcommand{\NM}{\pmb{N}}

\newcommand{\Lcal}{\mathcal{L}}

\newcommand{\Ccal}{\mathcal{C}}
\newcommand{\Fcal}{\mathcal{F}}
\newcommand{\Ucal}{\mathcal{U}}
\newcommand{\Tcal}{\mathcal{T}}
\newcommand{\Mcal}{\mathcal{M}}
\newcommand{\Xcal}{\mathcal{X}}
\newcommand{\Ical}{\mathcal{I}}
\newcommand{\Jcal}{\mathcal{J}}
\newcommand{\Qcal}{\mathcal{Q}}
\newcommand{\Kcal}{\mathcal{K}}

\newcommand{\rvVec}[1]{\pmb{\mathrm{#1}}}
\newcommand{\rvMat}[1]{\pmb{\mathsf{#1}}}
\newcommand{\nr}[1]{n_{\text{r},#1}}
\newcommand{\nt}{n_{\text{t}}}
\newcommand{\snr}{\mathsf{snr}}
\newcommand{\V}[2]{V_{#1\shortrightarrow#2}}

\newcommand{\Yhat}[2]{\hat{Y}_{#1\shortrightarrow#2}}
\newcommand{\rvVecYhat}[2]{\hat{\rvVec{Y}}_{#1\shortrightarrow#2}}
\newcommand{\Y}[2]{{Y}^{(#1)}_{#2}}

\newcommand{\vecy}[2]{{\pmb{y}}^{(#1)}_{#2}}
\newcommand{\M}[2]{M_{#1\shortrightarrow#2}}
\newcommand{\Mhat}[2]{\hat{M}_{#1\shortrightarrow#2}}
\newcommand{\R}[2]{R_{#1\shortrightarrow#2}}
\renewcommand{\v}[2]{v_{#1\shortrightarrow#2}}
\newcommand{\vecv}[2]{\pmb{v}_{#1\shortrightarrow#2}}
\newcommand{\yhat}[2]{\hat{y}_{#1\shortrightarrow#2}}
\newcommand{\vecyhat}[2]{\hat{\pmb{y}}_{#1\shortrightarrow#2}}
\newcommand{\m}[2]{m_{#1\shortrightarrow#2}}
\newcommand{\X}[1]{X^{(#1)}}
\newcommand{\rvVecX}[1]{\rvVec{X}^{(#1)}}
\newcommand{\x}[1]{x^{(#1)}}
\newcommand{\vecx}[1]{\pmb{x}^{(#1)}}
\newcommand{\Q}[1]{Q^{(#1)}}
\newcommand{\q}[1]{q^{(#1)}}
\newcommand{\vecq}[1]{\pmb{q}^{(#1)}}
\renewcommand{\S}[1]{S^{(#1)}}
\newcommand{\s}[1]{s^{(#1)}}
\newcommand{\vecs}[1]{\pmb{s}^{(#1)}}

\newcommand{\cond}{\,|\,}

\begin{document}

\title{A Novel Transmission Scheme for the $K$-user Broadcast Channel with Delayed CSIT}

\author{Chao~He,~\IEEEmembership{Student Member,}  Sheng~Yang,~\IEEEmembership{Member,} and~Pablo~Piantanida,~\IEEEmembership{Senior~Member}}%
\maketitle

\begin{abstract}
The state-dependent $K$-user memoryless Broadcast Channel~(BC) with state
feedback is investigated. We propose a novel transmission scheme and derive
its corresponding achievable rate region, which, compared to some general
schemes that deal with feedback, has the advantage of being relatively simple and thus
is easy to evaluate. In particular, 
it is shown that the capacity region of the symmetric erasure BC with an arbitrary input alphabet size is
achievable with the proposed scheme. For the fading Gaussian BC, we derive a
symmetric achievable rate as a function of the signal-to-noise ratio~(SNR) and a small set of parameters. 
Besides achieving the optimal degrees of freedom at high SNR, the proposed scheme is shown, through numerical
results, to outperform existing schemes from the literature in the finite SNR regime.
{\let\thefootnote\relax\footnote{The authors are with the Laboratoire des Signaux et Systèmes (L2S, UMR CNRS 8506),  CentraleSup\'elec - CNRS - Universit\'e Paris-Sud, 3, rue Joliot-Curie, 91190, Gif-sur-Yvette, France. (\{chao.he, sheng.yang, pablo.piantanida\}@centralesupelec.fr).}}
\end{abstract}


\section{Introduction}


With the dramatic growth of the number of mobile devices, modern wireless communication networks have become
interference limited. As such, the interference mitigation problem has attracted a surge of interest in recent
years. In a downlink Broadcast Channel~(BC), for instance, it is well known that interference can be
efficiently mitigated through precoding, provided that timely Channel State Information~(CSI) is available at
the transmitter side~(CSIT)~(see, e.g., \cite{Caire-Jindal-Kobayashi-Ravindran-TIT10} and the references therein).
While timely CSIT may not be available in mobile communications, it has been revealed in~\cite{wang2012capacity,gatzianas2013multiuser,maddah2012completely} that delayed CSIT is still very useful and can strictly enlarge the capacity region of a BC.  

In particular, the capacity region of the erasure BC (also referred
to as the EBC) with delayed CSIT was fully determined for up to three users and partially characterized for the
case with more users \cite{wang2012capacity, gatzianas2013multiuser}. The main idea behind their proposed schemes
in \cite{wang2012capacity} and \cite{gatzianas2013multiuser} is fundamentally the same: the source first sends out the source message packets, then
generates according to the state feedback some adequate \emph{linear combinations} of the packets that are erased at certain receivers but overheard
by some others. Such linear combinations are then multicast to a group of users in later phases. 
Their schemes are carefully designed such that at the end of the transmission a sufficient number of linearly
independent combinations are available to each receiver for the decoding of the original message packet. However, the schemes
in~\cite{wang2012capacity,gatzianas2013multiuser} are limited to packet erasure channels for which the input
alphabet size can only be $2^q$ with
$q\in\mathbb{N}$ being the number of bits per packet. In addition, there is an extra constraint, $2^q\ge K$, to guarantee the existence of a desired number of linearly independent vectors in the corresponding
vector space in finite field. As such, the capacity region is still open for the general EBC with arbitrary alphabet sizes.

For the multi-antenna fading Gaussian BC (also referred to as the GBC) with delayed CSIT, 
Maddah-Ali and Tse proposed a linear scheme that
achieves the optimal Degree of Freedom~(DoF) for the $K$-user Multiple-Input-Single-Output~(MISO) case. The authors
showed that with delayed CSIT the sum-DoF can still scale almost linearly with the number of users. Remarkably,
there is a striking similarity between the Maddah-Ali-Tse~(MAT) scheme and the
schemes from~\cite{wang2012capacity,gatzianas2013multiuser}. 
Namely, based on the CSI feedback, the transmitter can create and transmit useful linear combinations of the past received signals by the users. The intended group of users receive such linear combinations and use them to decode the message together with the previous observations.
Note that the MAT scheme in~\cite{maddah2012completely} has a fixed structure designed based 
 on a dimension counting argument. Although such a structure guarantees the DoF optimality at high
 Signal-to-Noise-Ratio~(SNR), it may not be efficient at finite SNR due to its inflexibility. 

 As a matter of fact, there are only a small number of works on the performance gain with delayed CSIT in the
 finite SNR regime. In~\cite{yi2013precoding}, the authors developed two linear precoding methods that attempt to 
 balance 
 the interference and the useful received signal. For $K=2$ and $3$, performance gain over MAT was revealed when a
 specific type of decoder is used. To adjust the multicast cost in the MAT scheme, the authors of
 \cite{ali2014approximate} proposed to transmit a quantized version of the linear combinations. For the
 $K$-user Rayleigh fading case, they
 demonstrated that a gap between the corresponding inner bound and a genie-aided outer bound, in terms of the
 symmetric rate, is upper bounded by $2\log_2 (K+2)$ which scales sublinearly with $K$. More recently,
 the work \cite{wangjue2013precoder} studied a scenario where both the CSI statistics and the feedback of
 the channel realizations are available at the transmitter. It was shown by numerical examples that statistics of CSI
 can enlarge the rate region for temporally correlated Rayleigh fading GBC. The authors of \cite{clerckx2015space}
 investigated the outage performance for GBC with an adapted MAT scheme. It is worth mentioning that these schemes
 are variants of the linear MAT scheme, i.e., they applied either linear coding or linear coding with quantization
 with the same fixed frame structure of MAT. 
 Although the rate performance of the MAT-like schemes is rather convincing in the medium-to-high SNR regime, their
 performance in the medium-to-low SNR regime is still questionable since it can be strictly dominated by the simple
 time-division multiple access~(TDMA) strategy~\cite{he2014capacity}. 
 
Instead of imposing the linear structure, we can tackle the problem directly from the
information-theoretic perspective. To that end, we formulate the setup as  
a $K$-user state-dependent memoryless BC with state feedback. This formulation includes both the EBC and
the GBC as special cases. In the two-user case, 
Shayevitz and Wigger studied such BC with generalized feedback and derived a general achievable
rate region using information-theoretic tools~\cite{shayevitz2013capacity}. Later on, Kim~\emph{et
al.}~demonstrated~in~\cite{kim2015note} that in the two-user symmetric setting, the Shayevitz-Wigger~(SW) region,
actually includes the MAT region. Similar recent works on the two-user case have been reported in
\cite{venkataramanan2013achievable, wu2016coding}. In this work, we are interested in the general $K$-user case. 
The main contributions are summarized as follows.
\begin{itemize}
\item We propose a novel scheme for the general $K$-user channel and derive the corresponding achievable rate region. The novelty of this scheme
  lies in the proper combination of two main ingredients: coded time-sharing and joint source-channel coding~(JSC) with side information at the
  decoder. We refer to our scheme in short as the JSC scheme.
  As compared to the existing schemes, e.g., the Shayevitz-Wigger scheme~(which is limited to two users)~\cite{shayevitz2013capacity},  
  our scheme is conceptually simpler in the sense that neither block-Markov coding nor Marton coding is required.  Such simplicity, at the cost
  of a slight loss of generality, allows us to derive the $K$-user rate region with a reasonable number of parameters. To the best of our
  knowledge, the JSC scheme is the first information-theoretic scheme for the $K$-user BC with state feedback for $K\ge3$. 
\item The general rate region is then evaluated for both the EBC and fading GBC. First, we show that our scheme achieves the capacity of a
  symmetric EBC with an arbitrary input alphabet size, whereas the previous schemes
  in~\cite{wang2012capacity,gatzianas2013multiuser} only apply to packet erasure channels. 
Second, for the symmetric fading GBC, we derive the achievable symmetric rate as a function of SNR and a set of $K-1$ compression noise variances. 
At high SNR, we show analytically that the proposed scheme achieves the optimal DoF under the same setting as in~\cite{maddah2012completely}. 
At finite SNR, we perform numerical optimization over the set of $K-1$ variances. The results show that 
in the two- and three-user cases, the JSC scheme outperforms the existing schemes at \emph{all} SNR. 
\end{itemize}

The remainder of the paper is organized as follows. We introduce the system model formally in
Section~\ref{ch-Kuser:general model}. Then we begin with the two-user case in Section~\ref{sec:two-user}, before
presenting the general $K$-user scheme in a more abstract way in Section~\ref{ch-Kuser:main result}. The general
region is applied to the erasure BC and fading Gaussian BC in Section~\ref{sec:application}. In
Section~\ref{ch-Kuser:sec:Simulation}, numerical results are provided for the two-user and
three-user fading Gaussian BC where we compare the JSC scheme to some baseline schemes from the literature. The
paper is concluded in Section~\ref{ch-Kuser:sec:Conclusion}. Although most of the derivations are provided in the
main text, some more technical details are deferred to the appendices.

\subsubsection*{Notation}
First, for random quantities, we use upper case letters, e.g., $X$, for scalars, upper case letters with bold and non-italic fonts,
e.g., $\rvVec{V}$, for vectors, and upper case letter with bold and sans serif fonts, e.g., $\rvMat{M}$, for matrices. 
Deterministic quantities are denoted in a rather conventional way with italic letters, e.g., a scalar $x$, a vector
$\pmb{v}$, and a matrix $\pmb{M}$. Logarithms are in base $2$.  
Calligraphic letters are used for sets. In particular, we let $\Kcal\defeq\left\{ 1,\ldots,K \right\}$ be the
set of all users. To denote subset of users, we use $\Ical$ and $\Jcal$ for some subsets with implicit size
constraints $|\Ical|=i$ and $|\Jcal|=j$, respectively. The constraints are
made explicit when necessary. $\Ucal$ is also used as subset of users but without size constraint. Hence,
$\{V_{\Ical}\}_{\Ical}\equiv \{V_\Ical: \Ical \subseteq \Kcal,
|\Ical| = i\}$ and $\{V_{\Ucal}\}_{\Ucal}\equiv \{V_\Ucal: \Ucal
\subseteq \Kcal\}$. 
The complement of $\Ical$ in $\Jcal$ is denoted by $\Jcal\setminus\Ical$.  We use $\bar{\Ucal}$ to denote the
complement of the set $\Ucal$ in $\Kcal$, i.e., $\bar{\Ucal} = \Kcal\setminus\Ucal$. 

\section{System Model}\label{ch-Kuser:sec:systemmodel}
\label{ch-Kuser:general model}

We consider a $K$-user state-dependent memoryless BC in
which the source wishes to communicate, in $n$ channel uses, $K$ independent messages to the
$K$ receivers, respectively. The channel can be described by the joint probability mass function~(pmf),
\begin{align}
  p(\pmb{y}_1,\ldots,\pmb{y}_K|\pmb{x},\pmb{s})p(\pmb{s}) =\prod_{i=1}^{n}
p(y_{1i},\ldots,y_{Ki}|x_i,s_i) p(s_i) \label{ch-Kuser:eq:system model}
\end{align}
where $\pmb{x}\in \mathcal{X}^n$, $\pmb{y}_k\in\mathcal{Y}_k^n$, and
$\pmb{s}\in\mathcal{S}^n$ are the sequences of the channel input, the channel output at the $k$-th
receiver, and the channel state, respectively. The channel state information~(CSI) is
known instantaneously to all the receivers. At transmitter's side, the channel state is known
strictly causally without error via a noiseless feedback link from the receivers. For simplicity, we
assume that the CSI is provided at the transmitter with one slot delay and the channel itself is
temporally i.i.d. The channel model is illustrated in Fig.~\ref{ch-Kuser:fig:systemmodel}.

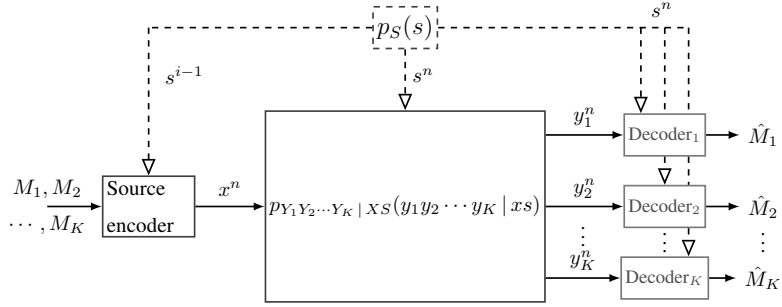
\begin{figure}[!t]
{\centering
\resizebox{0.65\textwidth}{!}{
\begin{tikzpicture}[scale=1,
enc/.style={rectangle,minimum size=7mm,draw=black!75,fill=white!30,thick,inner sep=2pt},
transp/.style={rectangle,minimum size=7mm,minimum height=32mm,draw=black!75,fill=white!30,thick,inner sep=2pt},
dec/.style={rectangle,minimum size=7mm,draw=black!75,fill=white!30,thick,inner sep=2pt,opacity = 0.7},
state/.style={rectangle,minimum size=7mm,draw=black!75,fill=white!30,thick,inner sep=2pt,dashed}
]
\node at (-6,0.3) (msg) [font=\small]{$M_1,M_2$};
\node at (-6,-0.3) (msg) [font=\small]{$\cdots,M_K$};
\node at (-4.3,0) (src) [enc,font=\small,text width=14mm]{Source encoder};
\node at (0,0) (prob) [transp,font=\small]{$p_{Y_1 Y_2\cdots Y_K \cond X S}(y_1 y_2 \cdots y_K \cond x s)$};
\node at (4.35,1.2) (dec1) [dec,font=\footnotesize]{$\textrm{Decoder}_1$};
\node at (4.35,0) (dec2) [dec,font=\footnotesize]{$\textrm{Decoder}_2$};
\node at (4.35,-0.5) (decvdot) [font=\small]{$\vdots$};
\node at (4.35,-1.2) (decK) [dec,font=\footnotesize]{$\textrm{Decoder}_K$};
\node at (0,3) (state) [state]{$p_S(s)$};
\node at (6,1.2) (estmsg1) [font=\small]{$\hat{M}_1$};
\node at (6,0) (estmsg2) [font=\small]{$\hat{M}_2$};
\node at (6,-0.5) (estmsgvdot) [font=\small]{$\vdots$};
\node at (6,-1.2) (estmsgK) [font=\small]{$\hat{M}_K$};

\node at (-3.7,2.2) (msg) [font=\small]{$s^{i-1}$};
\node at (0.3,2.2) (msg) [font=\small]{$s^n$};
\node at (4.3,3.3) (msg) [font=\small]{$s^n$};

\node at ($(dec1.north)+(0.4,0)$) (phantomnode1) [] {};
\node at ($(dec1.north)+(-0.4,0)$) (phantomnode2) [] {};

\node at ($0.5*(src.east)+0.5*(prob.west)+(0,0.3)$) (input) [font=\small] {$x^n$};
\node at ($0.5*(dec1.west)+0.5*(prob.east|-dec1.west)+(0,0.3)$) (output1) [font=\small] {$y_1^n$};
\node at ($0.5*(dec2.west)+0.5*(prob.east|-dec2.west)+(0,0.3)$) (output2) [font=\small] {$y_2^n$};
\node at ($0.5*(decK.west)+0.5*(prob.east|-decK.west)+(0,0.8)$) (outputvdot) [font=\small]{$\vdots$};
\node at ($0.5*(decK.west)+0.5*(prob.east|-decK.west)+(0,0.3)$) (outputK) [font=\small] {$y_K^n$};

\draw [-latex,thick] (-6,0) to (src);
\draw [-latex,thick] (src) to (prob);
\draw [-latex,thick] (prob.east|-dec1.west) to (dec1.west);
\draw [-latex,thick] (prob.east|-dec2.west) to (dec2.west);
\draw [-latex,thick] (prob.east|-decK.west) to (decK.west);
\draw [-latex,thick] (dec1.east) to (estmsg1.west);
\draw [-latex,thick] (dec2.east) to (estmsg2.west);
\draw [-latex,thick] (decK.east) to (estmsgK.west);
\draw [-open triangle 45,thick, dashed] (state) to (prob);
\draw [-,thick, dashed] (state.west) to (state.west-|src.north);
\draw [-open triangle 45,thick, dashed] (state.west-|src.north) to (src.north);
\draw [-,thick, dashed] (state.east) to (state.east-|phantomnode1.north);

\draw [-open triangle 45,thick, dashed] (state.east-|phantomnode2.north) to (phantomnode2.north|-dec1.north);
\draw [-,thick, dashed] (state.east-|dec1.north) to (dec1.north);
\draw [-,thick, dashed] (state.east-|phantomnode1.north) to (phantomnode1.north|-dec1.north);

\draw [-open triangle 45,thick, dashed] (dec1.south) to (dec2.north);
\draw [-,thick, dashed] (dec1.south-|phantomnode1.north) to (dec2.north-|phantomnode1.north);

\draw [-open triangle 45,thick, dashed] (dec2.south-|phantomnode1.north) to (decK.north-|phantomnode1.north);
%
%
\end{tikzpicture}
}
\caption{General system model of $K$-user BC with state feedback.}
\label{ch-Kuser:fig:systemmodel}
}
\end{figure}

Let the message for user~$k$, $M_k$, be uniformly distributed in the 
 message set $\mathcal{M}_k\defeq [1:2^{n {R}_k}]$, for $k\in\Kcal$. We say that the rate tuple
$(R_1,\ldots,R_K)$ is achievable if there exist 
\begin{itemize}
\item a sequence of encoding
functions $\{f_i:\ \mathcal{M}_1\times\cdots\times\mathcal{M}_K\times
\mathcal{S}^{i-1} \to \mathcal{X}\}_{i=1}^n$, and
\item $K$ decoding functions
$\{g_{k}:\ \mathcal{Y}_k^n\times\mathcal{S}^n \to
\mathcal{M}_k\}_{k=1}^K$,
\end{itemize}
 such that 
$\max_k \Pr\big(g_{k}(Y_k^n, S^n) \ne M_k\big) \to 0$ when $n\to\infty$. The symmetric rate
$R_{\text{sym}}$ is achievable if the rate tuple $(R_{\text{sym}}, \ldots,
R_{\text{sym}})$ is achievable. 
In particular, we are interested in the following two specific channels. 


\subsection{Fading Gaussian Broadcast Channel}

The fading GBC with $n_t$ transmit antennas and $n_{r,k}$ receive antennas at user $k$, $k\in\Kcal$, is defined by
\begin{align}
  \rvVec{Y}_k &= \rvMat{H}_k \rvVec{X} + \rvVec{Z}_k, \quad k\in\Kcal,
\label{ch-Kuser:eq:Gaussiansystemmodel}            
\end{align}%
where $\rvVec{X} \in \mathbb{C}^{\nt\times1}$ is the input vector, $\rvVec{Y}_k \in \mathbb{C}^{\nr{k}\times1}$ is
the output vector at receiver~$k$, $\rvVec{Z}_k \sim \mathcal{CN}(\pmb{0}, \sigma^2 \IM_{\nr{k}})$ is the additive
white Gaussian noise~(AWGN), and $\rvMat{H}_k \in \mathbb{C}^{\nr{k}\times\nt}$ is the channel matrix to
receiver~$k$. The channel input is subject to the power constraint as $\frac{1}{n}\sum_{i=1}^{n} \| \pmb{x}_i \|^2
\leq P$ for any input sequence $\pmb{x}_1,\ldots,\pmb{x}_n$. The SNR is defined as $\snr \defeq
\frac{P}{\nt\sigma^2}$. We assume that both the channel matrices and the AWGN are independent across users. 
We use $\rvMat{H}_{\Ucal}$ to denote a matrix from a vertical concatenation of the channel matrices of receivers in
$\Ucal$, i.e., $\{\rvMat{H}_k\}_{k\in \Ucal}$, same notation applies for $\rvVec{Y}_{\Ucal}$ and $\rvVec{Z}_{\Ucal}$. Hence, it follows that
$\rvVec{Y}_{\Ucal} = \rvMat{H}_{\Ucal} \rvVec{X} + \rvVec{Z}_{\Ucal}$.
The matrix $\rvMat{H}_{\Kcal}$ corresponds to the channel state $S$ in the general formulation. 

\subsection{Erasure Broadcast Channel}\label{ch-Kuser:subsec:EBCsystemmodel}
The EBC is a state-dependent deterministic channel in which
\begin{align}
  Y_k &= \begin{cases} X, & S_k = 1, \\
    \,?\,, & S_k = 0, \end{cases}
\end{align}%
for $k\in\Kcal$. Here the input alphabet $\Xcal$ is arbitrary and finite with size $|\Xcal|$; 
, ``$?$'' stands for erasure, and the output alphabet is $\mathcal{Y}=\Xcal\cup\left\{ ? \right\}$. 
The distribution of the channel state is characterized by the set of probabilities
\begin{align}
  \P_{\Ucal, \bar{\Ucal}} &\defeq \Pr(S_{\Ucal} = \pmb{0},\ S_{\bar{\Ucal}} =
  \pmb{1}), \quad \Ucal \subseteq \left\{ 1,\ldots,K \right\}
\end{align}%
with $\sum_{\Ucal} \P_{\Ucal,\bar{\Ucal}} = 1$. Throughout the paper, we use
$S_{\Ucal} = \pmb{0}$~(\emph{resp.} $S_{\Ucal} = \pmb{1}$) to define the event that
$S_k=0$~(\emph{resp.} $S_k = 1$), $\forall~k\in \Ucal$. For simplicity, we use $\delta_{\Fcal}$ to
denote $\Pr(S_{\Fcal}=\pmb{0})$, and use $\P_{\Fcal,\Tcal}$ to denote
$\Pr(S_{\Fcal}=\pmb{0},S_{\Tcal}=\pmb{1})$  for any $\Fcal$ and $\Tcal$ that satisfy $\Fcal \cap
\Tcal = \phi $ and $\Fcal \cup \Tcal \subseteq \Kcal$.
For notational brevity, $\delta_{\{k\}}$ is written as $\delta_k$. 
The $K$-tuple $S_{\Kcal}$ corresponds to the channel state $S$ in the general formulation.

\section{The Two-User Case}
\label{sec:two-user}

Before presenting the main results for the general $K$-user channel, we provide a description of the two-user case.
The goal is to explain the main ingredients of the proposed scheme in an accessible and less formal way, whereas
a rigorous and detailed description will be provided for the $K$-user case in the next section. Hereafter, we
also refer to our scheme as the JSC scheme. 

\subsection{Scheme description}

The JSC scheme consists in two phases with $n_1$ and $n_2$ being the length of phase~1 and 2, respectively. 
The total transmission length is $n=n_1+n_2$.
In the first phase, the original messages $M_1$ and $M_2$, for receiver~1 and receiver~2, respectively,
are encoded and transmitted. At the end of phase~1, the transmitter obtains the state feedback and thus some side information about the received
signals at both users during phase 1. In the second phase, the transmitter compresses the side information that are useful to both users into
$\hat{Y}$, and transmits the compression index $M_{12}$ with a channel code. Each receiver decodes the compression index first, but with the
observation from \emph{both} phases. With the compression index, the side information $\hat{Y}$ is recovered and combined with the observation
from phase~1 by each receiver $k\in\left\{ 1,2 \right\}$ to finally decode the message $M_k$. 
The main information-theoretic tools that we use in this scheme are the following ones:
\begin{itemize}
  \item Coded time-sharing for the transmission in phase~1;
  \item Joint source-channel coding in phase~2;
  \item Joint source-channel decoding with side information on the source. 
\end{itemize}

\subsection*{\underline{Phase~1}}
At the beginning, we randomly generate a sequence of time-sharing variables $\pmb{q} \defeq (q_1,\ldots,q_{n_1})$
according to $\prod_{t=1}^{n_1}p(q_t)$. For each user $k\in\left\{ 1,2 \right\}$, we generate a random codebook of
$2^{n R_k}$ independent codewords, $\pmb{v}_k(m_k)$, $m_k \in [1:2^{n R_k}]$, each according to $\prod_{t=1}^{n_1}
p(v_{k,t})$. Both the time-sharing sequence and the codebooks are revealed to the transmitter and all the
receivers. 

To send messages $m_1$ and $m_2$ to user~1 and 2, respectively, we use coded time-sharing in phase 1. Specifically, at time $t$, the transmitter
sends $x^{(1)}_t = v_{1,t}$ if $q=1$ and $x^{(1)}_t = v_{2,t}$ if $q=2$. It is similar to a TDMA scheme controlled by the time-sharing variables
$\{q_t\}$. Here the superscript $\cdot^{(1)}$ stands for phase~1. 

At the end of phase~1, each receiver~$k$ observes $y_{k,t}$ that depends on $x_t$ and the channel state $s_t$, for $t=1,\ldots,n_1$. The
transmitter obtains through feedback the sequence $s_1,\ldots,s_{n_1}$. At this point, the transmitter knows the following i.i.d.~triples
\begin{align}
  (v_{1,1}, v_{2,1}, s_1), \ldots, (v_{1,n_1}, v_{2,n_1}, s_{n_1}),  \label{eq:tmp99}
\end{align}%
which can be regarded as a source sequence of length $n_1$. Alternatively, it can be represented by $(\pmb{v}_1, \pmb{v}_2, \pmb{s}^{(1)})$.

\subsection*{\underline{Phase~2}}
The transmitter creates a source codebook and a channel codebook, both with the same size~$2^{n_1 R_{12}}$. Specifically, the source codebook
contains $2^{n_1 R_{12}}$ i.i.d.~sequences $\hat{\pmb{y}}(m_{12})$, each generated according to $\prod_{t=1}^{n_1} p(\hat{y}_t \cond s_t, q_t)$,
whereas the channel codebook contains $2^{n_1 R_{12}}$ sequences ${\pmb{v}_{12}}(m_{12})$, each generated according to $\prod_{t=1}^{n_2}
p(v_{12,t})$. Note that while the codebook size is the same, the codeword lengths are different for the source and channel codebooks. This is
because the source codebook is used to describe the source sequence \eqref{eq:tmp99} from phase~1 while the channel codebook is used to send the
index $m_{12}$ in phase~2. 

First, the transmitter finds a sequence $\hat{\pmb{y}}(m_{12})$ from the source codebook that is jointly typical with the source
sequence~\eqref{eq:tmp99}. This can be done successfully provided that
\begin{align}
  R_{12} \geq  I(\hat{Y};V_{1},V_2 \cond \S{1},Q). \label{eq:tmp1} 
  \end{align}%
Then, the source sequence is associated with the channel codeword ${\pmb{v}_{12}}(m_{12})$ through the
index~$m_{12}$. The transmission in phase~2 is simply specified by $x^{(2)}_t = v_{12,t}$, $t=1,\ldots,n_2$. The above procedure can be seen as a joint source-channel coding. 

\subsection*{\underline{Decoding}}
We focus on the decoding at receiver~$k$ without loss of generality. First, the receiver tries to find out $m_{12}$ with the observations from the
two phases: $\pmb{y}^{(1)}_k$ and $\pmb{y}^{(2)}_k$. Intuitively, $\pmb{y}^{(1)}$ is correlated with the source $(\pmb{v}_1, \pmb{v}_2,
\pmb{s}^{(1)})$ and thus the source codeword $\hat{\pmb{y}}(m_{12})$, whereas $\pmb{y}^{(2)}_k$ is correlated with the channel codeword
${\pmb{v}_{12}}(m_{12})$. Hence, both observations can help find the \emph{same} index $m_{12}$. Specifically,  
decoder~$k$ looks for $\hat{m}_{12}$ such that $\bigl(\pmb{v}_{12}(\hat{m}_{12}), \vecy{2}{k}, \vecs{2} \bigr)$ are jointly typical and that
$\bigl(\hat{\pmb{y}}(\hat{m}_{12}),\vecy{1}{k},\vecs{1},\pmb{q}\bigr)$ are jointly typical. 
It turns out that one can recover $m_{12}$ correctly as long as the rate satisfies 
\begin{align}
  n_1 R_{12}  \leq   n_{1} I(\hat{Y};
  \Y{1}{k} \cond \S{1}, Q) + n_2 I(V_{12}; \Y{2}{k} \cond \S{2}),\label{ch-2user:eq:decodingcondition}
\end{align}
where we see clearly the contribution of the observations from both phases. This is essentially Tuncel's scheme~\cite{tuncel2006slepian} of separate source-channel encoding but joint source-channel decoding. 

Then the receiver uses $\hat{\pmb{y}}(\hat{m}_{12})$ jointly with the observation from phase~1, $\pmb{y}^{(1)}$, to decode the original message.
Specifically, it looks for the unique $\hat{m}_k$ such that $(\pmb{v}_k(\hat{m}_k), \pmb{y}_k^{(1)}, \hat{\pmb{y}}(\hat{m}_{12}), \pmb{q},
\pmb{s}^{(1)})$ is jointly typical. The original message can be decoded correctly if the message rate satisfies  
   \begin{align}
     n R_{k} \leq n_1 I(V_{k}; \Y{1}{k}, \hat{Y}\cond \S{1}, Q). \label{eq:tmp3}
  \end{align}

  From \eqref{eq:tmp1}-\eqref{eq:tmp3}, we see that for any fixed distribution $p(v_1)p(v_2)p(v_{12})p(\hat{y}\cond v_1, v_2, s^{(1)}, q)$ 
  the rate pair $(R_1,R_2)$ is achievable if, for $k\in\left\{ 1,2 \right\}$, 
  \begin{align}
    R_{k} &\leq \alpha_1 I(V_{k}; \Y{1}{k}, \hat{Y}\cond \S{1}, Q), \label{eq:rk} \\
I(\hat{Y};V_{1},V_2 \cond \S{1},Q) &\le I(\hat{Y}; \Y{1}{k} \cond \S{1}, Q) + \frac{\alpha_2}{\alpha_1} I(V_{12}; \Y{2}{k} \cond \S{2}), \label{eq:tradeoff}
  \end{align}%
  where $\alpha_1 \defeq \frac{n_1}{n}$ and $\alpha_2\defeq \frac{n_2}{n}$ with $\alpha_1+\alpha_2=1$ can be optimized.

  \begin{remark}[Comparison to other information-theoretic schemes]
  
  Although in the two-user case our result is closely connected to the works \cite{shayevitz2013capacity} and \cite{venkataramanan2013achievable}, unfortunately it is hard to make a fair comparison in the Gaussian case. 
  First, the fact that the achievable regions depend on different sets of pmf's prohibits the analytical comparison. Then, since we cannot find the
  exact optimal solution for any of the regions~(e.g. Gaussian input is not even proved to be optimal in general), any numerical comparison must
  be based on a particular choice of distribution, which cannot be conclusive. Indeed, the underlying transmission schemes are conceptually
  different. Both schemes in \cite{shayevitz2013capacity} and \cite{venkataramanan2013achievable} use Marton coding and block-Markov coding. While
  a separate source channel coding was used to compress the side information in \cite{shayevitz2013capacity}, the
  authors in \cite{venkataramanan2013achievable}
  adopted a joint source-channel coding approach. In our scheme, we do not use Marton coding nor block-Markov
  coding, but, as in
  \cite{venkataramanan2013achievable}, we use a joint source-channel coding for the transmission of side
  information. It is worth noting that our scheme is based on
  Tuncel's scheme \cite{tuncel2006slepian} and is different from the one used in \cite{venkataramanan2013achievable}. 
 
  From the complexity perspective, our scheme is conceptually simpler since neither block-Markov coding nor binning is required. Furthermore, due to the relative simplicity, we
  manage to derive a rate region for the $K$-user case with a reasonable number of parameters, as will be shown in the upcoming sections. Such
  an advantage allows us to easily obtain numerical results for $K \ge 3$, whereas the counterpart of the existing
  two-user schemes is still missing in the literature due to their complexity for extension.
\end{remark}

  \subsection{Application to the fading GBC}

  Let us consider a two-user MISO fading GBC with $Y_k^{(j)} = \rvVec{H}_k^{(j)H} \rvVec{X}^{(j)} + Z_k^{(j)}$, for
  $k\in\left\{ 1,2 \right\}$ and $j\in\left\{ 1,2 \right\}$. Here $\rvVec{H}_k^{(j)H}$ is the channel
  $\rvMat{H}_k^{(j)}$ in the vector case. We consider the symmetric case and let $Q \in \left\{ 1,2 \right\}$ be uniform with probability~$\frac{1}{2}$ for
  each user. In phase~1, we have $\rvVec{X}^{(1)} = \rvVec{V}_1$ if $Q=1$ and $\rvVec{X}^{(1)} = \rvVec{V}_2$ if $Q=2$, with
  $\rvVec{V}_1,\rvVec{V}_2 \sim \mathcal{CN}(0, \frac{P}{2} \Id_2)$ being independent. At the end of phase~1, the transmitter sets the side information as follows
  \begin{align}
    \hat{Y} &= \begin{cases} \rvVec{H}_2^{(1)H} \rvVec{X}^{(1)} + \hat{Z}_2 = \rvVec{H}_2^{(1)H} \rvVec{V}_1 + \hat{Z}_2, & \text{if } Q = 1, \\
    \rvVec{H}_1^{(1)H} \rvVec{X}^{(1)} + \hat{Z}_1 = \rvVec{H}_1^{(1)H} \rvVec{V}_2 + \hat{Z}_1, & \text{if }  Q = 2. 
    \end{cases} \label{eq:Yhat}
  \end{align}%
  Intuitively, $\hat{Y}$ is the compression of the \emph{overheard} signal at the \emph{unintended} receiver, with
  $\hat{Z}_1, \hat{Z}_2\sim\mathcal{CN}(0,\hat{\sigma}^2)$ being the independent compression noises. 
  The idea is not to retransmit everything about $(\rvVec{V}_1,
  \rvVec{V}_2, \rvMat{H})$ as side information, since this would be too costly. Instead, sending compressed version of a function of these
  information, namely, $(\rvVec{H}_2^{(1)H} \rvVec{V}_1,  \rvVec{H}_1^{(1)H} \rvVec{V}_2)$, would be helpful. The compression noise
  can balance the precision of side information and the cost for the transmission, as can be observed in \eqref{eq:tradeoff}. To get more insight
  on the choice of $\hat{Y}$, let us rewrite \eqref{eq:tradeoff} as 
  \begin{align}
    I(\hat{Y};V_{1},V_2 \cond Y_k^{(1)}, \S{1},Q) &\le \frac{\alpha_2}{\alpha_1} I(V_{12}; \Y{2}{k} \cond \S{2}),
\label{eq:tradeoff2}
  \end{align}%
  using the Markovity $\hat{Y} \leftrightarrow (V_{1},V_{2}, \S{1},Q) \leftrightarrow  Y_k^{(1)}$ and the chain rule. The left-hand side of
  \eqref{eq:tradeoff2} is the average amount of side information remaining in $\hat{Y}$ after observing $Y_k^{(1)}$, whereas the right hand side
  is the achievable transmission rate to user~$k$ in phase~$2$. With the choice~\eqref{eq:Yhat}, on the one hand, we make sure that during half of
  the time, $I(\hat{Y};V_{1},V_2 \cond Y_k^{(1)}, \S{1},Q = q)$ is small since $Y_k^{(1)}$ already almost contains the information in $\hat{Y}$.
  This makes sure that the constraint \eqref{eq:tradeoff2} can be met. On the other hand, from \eqref{eq:Yhat} and \eqref{eq:rk}, we notice
  that $\hat{Y}$ provides to receiver~$k$ an extra observation approximative to the other receiver's signal. Such
  observation helps create a virtual MIMO system for
  each user. From \eqref{eq:rk}, it readily follows that the symmetric rate is
  \begin{align}
    R_\text{sym}^{\text{JSC}} &= \frac{\alpha_1}{2}\, \mathbb{E} \left[ \log \det \left(\Id + \left[ \begin{smallmatrix}
      \sigma^{-2} & 0 \\ 0 & \hat{\sigma}^{-2} \end{smallmatrix}\right] \frac{P}{2}
      \rvMat{H} \rvMat{H}^H \right) \right], \label{eq:rsym}
  \end{align}%
  where $\alpha_1$ should be chosen to satisfy \eqref{eq:tradeoff2} with equality,\footnote{If
  the inequality \eqref{eq:tradeoff2} is strict, then one can always increase $\alpha_1$ to achieves
  equality. This is without loss of optimality since increasing $\alpha_1$ only increases the
  symmetric rate.} that is,
  \begin{align}
    \frac{1}{2}\sum_{l=1}^2 \mathbb{E} \left[ \log\left( 1 + \frac{P}{2\hat{\sigma}^2} \rvVec{H}_l
    \left(1+\frac{P}{2\sigma^2} \rvVec{H}_k^{H}\rvVec{H}_k\right)^{-1}  \rvVec{H}_l^{H}
    \right)\right]  &= \frac{\alpha_2}{\alpha_1} \mathbb{E} \left[ \log\left( 1 + \frac{P}{2\sigma^2}\|\rvVec{H}_k
    \|^2   \right) \right].
  \end{align}%
  Combining the above equation with $\alpha_1+\alpha_2 = 1$, we obtain the solution
  \begin{align}
    \alpha_1 = \frac{\mathbb{E} \left[ \log\left( 1 + \frac{P}{2\sigma^2}\|\rvVec{H}_k
    \|^2   \right) \right] }{\mathbb{E} \left[ \log\left( 1 + \frac{P}{2\sigma^2}\|\rvVec{H}_k
    \|^2   \right) \right] + 
\frac{1}{2}\sum_{l=1}^2 \mathbb{E} \left[ \log\left( 1 + \frac{P}{2\hat{\sigma}^2} \rvVec{H}_l
    \left(1+\frac{P}{2\sigma^2} \rvVec{H}_k^{H}\rvVec{H}_k\right)^{-1}  \rvVec{H}_l^{H}
    \right) \right] }. \label{eq:alpha}  
  \end{align}%

  \subsection*{Comparison to the MAT scheme}

  The proposed scheme can be compared to the original 
  two-user MAT scheme which also works in two phases. There are three slots in total: two slots in
  phase~1 and one slot in phase~2. First, $\rvVec{V}_1$ and $\rvVec{V}_2$ are sent in slot~1
  and~2, respectively, in a TDMA fashion. At the end of phase~1, the transmitter receives the CSI
  feedback and linearly combines the \emph{overheard} signal from phase~1 as $L \defeq \rvVec{H}_2^H
  \rvVec{V}_1 + \rvVec{H}_1^H \rvVec{V}_2$. In phase~2, the symbol $L$ is scaled, e.g., to
  $\frac{1}{\sqrt{2}} L$ if the average transmit power constraint is imposed and if the channel coefficients are
  normalized. Then, the scaled signal is transmitted in one
  slot using one antenna. At the end, user~1 receives the noisy versions $\rvVec{H}_1^H
  \rvVec{V}_1 + Z_{11}$, $\rvVec{H}_1^H \rvVec{V}_2 + Z_{12}$, and $\frac{H_{11}'}{\sqrt{2}}(\rvVec{H}_2^H
  \rvVec{V}_1 + \rvVec{H}_1^H \rvVec{V}_2) + Z_{13}$
  in three slots, where $H_{11}'$ is the channel from the first antenna to user~1 at slot~3. From the three observations, receiver~1 gets the following virtual MIMO output
  \begin{align}
    Y_{11}' &=  \rvVec{H}_1^H \rvVec{V}_1 + Z_{11} \\
    Y_{12}' &=  \frac{1}{\sqrt{2}} H_{11}' \rvVec{H}_2^H \rvVec{V}_1 - \frac{1}{\sqrt{2}} H_{11}' Z_{12} + Z_{13}, 
  \end{align}%
  Due to the symmetry, receiver~2 has the similar form on $\rvVec{V}_2$. Finally, we conclude that
  the symmetric MAT rate is
  \begin{align}
    R_{\text{sym}}^{\text{MAT}} &= \frac{1}{3} \mathbb{E} \left[ \log \det \left(\Id + \left[ \begin{smallmatrix}
      \sigma^{-2} & 0 \\ 0 & \tilde{\sigma}_H^{-2} \end{smallmatrix}\right] \frac{P}{2}
      \rvMat{H} \rvMat{H}^H \right) \right], \label{eq:rsym-mat}
  \end{align}%
  where $\tilde{\sigma}_H^2 \defeq {\sigma}^2 (1 + \frac{2}{|H'_{11}|^2})$. Comparing \eqref{eq:rsym} and \eqref{eq:rsym-mat}, we notice the
  similarity of the rate expressions. Indeed, from \eqref{eq:alpha}, we see that if we set $\hat{\sigma}^2 =
  \sigma^2$, then $\alpha_1
  \xrightarrow{P\to\infty} \frac{2}{3}$ which implies that both the MAT and the JSC schemes have the same prelog
  factor and thus the same DoF. However,
  the noise covariance inside the determinant is different in \eqref{eq:rsym} and \eqref{eq:rsym-mat} since $\tilde{\sigma}_H^2$ is almost triple
  of $\hat{\sigma}^2$ if we approximate $|H_{11}'|^2$ by $1$. Although the above comparison is not precise, it provides an idea that the power gain
  of the JSC scheme over the MAT scheme at high SNR is mainly due
  to the fact that the linear operations in MAT cause cumulation of noises from different phases. Intuitively, it is analogous to the advantage of
  compress-forward like schemes over amplify-forward like schemes in relay channels. At finite SNR, the proposed JSC scheme also provides 
  the flexibility of choosing an appropriate compression noise variance $\hat{\sigma}^2$ as a function of $P$. Obviously, this flexibility
  requires that $\alpha_1$ in \eqref{eq:alpha} can be changed accordingly. In other words, with the JSC scheme one can adjust the length of the
  phases to achieve better performance, which is essential at finite SNR. Such flexibility is not possible with the MAT scheme since the length of
  each phase is fixed. Therefore, although the MAT scheme is DoF optimal, it may suffer from rate loss at finite
  SNR. More comments on the differences between the JSC scheme and the MAT-like schemes shall be made in
  Section~\ref{ch-Kuser:sec:Simulation}.


\section{The General Case with $K$ Users}\label{ch-Kuser:main result}
In this section, we describe the general JSC scheme for the $K$-user BC with state feedback. The rate region is given in the following main result of
this paper. 

\begin{theorem}\label{ch-Kuser:IBGBC} 
  A rate tuple $(R_1,\ldots,R_K)$
  is achievable in the $K$-user BC with causal state
  feedback if  
  \begin{align} 
    {R}_{k} &\leq \alpha_1 I(V_{k}; \Y{1}{k},
    \{\Yhat{1}{\Ucal}\}_{\Ucal \ni k} \cond \S{1},
    \Q{1}),\label{ch-Kuser:eq:IBofGBC2} \\ 
    0 &\le \min_{i,j,k,\Jcal:\atop i<j, k\in \Jcal} \Bigl\{ \alpha_j
    I(\V{i}{\Jcal};\Y{j}{k},\{\Yhat{j}{\Ucal}\}_{ \Ucal \supset \Jcal}
    \cond \S{j}, \Q{j})  - \alpha_i
    I(\{V_{\Ical}\}_{\Ical \subset \Jcal}; \Yhat{i}{\Jcal} \cond \Y{i}{k},
    \S{i}, \Q{i})  \Bigr\}.  \label{ch-Kuser:eq:IBofGBC1}
  \end{align} for some $K$-tuple
  $(\alpha_1,\ldots,\alpha_K)\in\mathbb{R}^K_+$ with $\sum_k \alpha_k =
  1$, and some pmf~\footnote{We define $v_{\Ucal} \defeq \left\{
    \v{k}{\Ucal}:\ k<|\Ucal| \right\}$ and $V_{\Ucal} \defeq \left\{
    \V{k}{\Ucal}:\ k<|\Ucal| \right\}$ for brevity. We also recall that $\Ical$ and $\Jcal$ are subsets of size $i$ and $j$, respectively.  
    }
  \begin{align}
      \left( \prod_{j=1}^K p(\x{j} \cond \{v_{\Jcal}\}_{\Jcal}, \q{j}) \right) \prod_{k=1}^K p(v_k)\prod_{j=2}^K 
      \prod_{\Jcal} \prod_{i=1}^{j-1} p(\v{i}{\Jcal}) p(\yhat{i}{\Jcal} \cond \{
      v_{\Ical}\}_{\Ical\subset \Jcal}, \s{i}, \q{i}),  \label{ch-Kuser:eq:pmf1} 
    \end{align}%
    \end{theorem} 
  If we let $K = 2$ and $Q^{(2)}$ be deterministic, and identify $(V_{1\to{12}}$, $\hat{Y}_{1\to{12}})$ with
  $(V_{{12}},\hat{Y})$, then we recover
  the results \eqref{eq:rk} and \eqref{eq:tradeoff2} in the two-user case. In order to have a general scheme for the $K$-user case, however, we
  need a much heavier notation as shown in the above theorem. As it will become clearer later, the complexity originates from the need to
  introduce $K$ phases in each one of which different types of information are created and sent. In particular, as compared to the two-user case,
  there are in general a coded time-sharing random variable~(RV) $\Q{j}$ in each phase~$j\in\Kcal$. The subscript ``$i\to \Jcal$'' can be understood as related to
  the side information created in phase~$i$ and intended for the users in group $\Jcal$ of $j$ users. We recall that in the two-user case, there
  is only one type of side information that is intended for both users~($|\Jcal| = 2)$. 


In the rest of the section, we present the general JSC scheme in detail and prove the achievability of the rate region given by Theorem~\ref{ch-Kuser:IBGBC}. We divide the $n$-slot transmission into $K$ phases, each phase $j\in\Kcal$ having length $n_j$ such that $n = n_1+\cdots+n_K$.  
We define the normalized length of phase $j$ as $\alpha_j\defeq\frac{n_j}{n}$ with $\sum_{j=1}^K \alpha_j=1$. 

The $K$-user scheme works in a similar manner as the two-user scheme. In each phase~$j\in\Kcal$, 
\begin{itemize}
  \item the input, output, and state are denoted by ${X}^{(j)}$, ${Y}^{(j)}$, and $S^{(j)}$, respectively;
  \item if $j=1$, the original messages $\{M_k \in \Mcal_{k} \defeq [1:2^{n R_k}]:\ k\in\Kcal\}$ are sent, otherwise a set of side
    information messages $\{\M{i}{\Jcal} \in \Mcal_{i\shortrightarrow\Jcal}\defeq [1:2^{n_i \R{i}{\Jcal}}]:\ |\Jcal|=j,\, i<j \}$ are sent;
  \item each message $\M{i}{\Jcal}$ is related to the source RV $\hat{Y}_{i\shortrightarrow\Jcal}$ created in a previous phase~$i$, and is carried
    by $V_{i\shortrightarrow\Jcal}$ for transmission;
  \item the transmission is controlled by a coded time-sharing RV $\Q{j}$.
\end{itemize}




\subsubsection*{\underline{Codebook Generation}}
Fix the pmf as described in~\eqref{ch-Kuser:eq:pmf1}. 
\begin{enumerate}
  \item Before the beginning of phase $j\in\Kcal$, randomly generate the time-sharing sequence according to $\prod_{t=1}^{n_j} p(\q{j}_t)$.
  \item At the beginning of phase $1$, for each user~$k\in\Kcal$, randomly generate $2^{n R_{k}}$ independent sequences $\pmb{v}_{k}(m_k)$, $m_{k} \in [1:2^{n R_{k}}]$, each according to $  \prod_{t=1}^{n_1} p(v_{k,t})$. 
  \item At the end of phase $i\in \{1,\ldots,K-1\}$, for each $j>i$ and each $\Jcal$ with $|\Jcal|=j$, randomly generate $2^{n_i \R{i}{\Jcal}}$
    independent sequences $\vecyhat{i}{\Jcal}(\m{i}{\Jcal})$ and $2^{n_i \R{i}{\Jcal}}$ independent sequences $\vecv{i}{\Jcal}(\m{i}{\Jcal})$,
    $\m{i}{\Jcal} \in[1:2^{n_i \R{i}{\Jcal}}]$, each according to $ \prod_{t=1}^{n_i} p(\yhat{i}{\Jcal,t} \cond \s{i}_t, \q{i}_t)$ and
    $\prod_{t=1}^{n_j} p(\v{i}{\Jcal,t})$, respectively. 
\end{enumerate}
\subsubsection*{\underline{Encoding}}
\begin{enumerate}
  \item In phase $1$, to send the original messages $(M_1,\ldots,M_K)$, a sequence $\vecx{1}$ is first generated 
    based on ($\pmb{v}_1(M_1), \cdots, \pmb{v}_K(M_K), \vecq{1}$) according to $ \prod_{t=1}^{n_1} p(x^{(1)}\cond v_{1,t},\ldots,v_{K,t}, q^{(1)}_t)$ and then transmitted.
\item At the end of phase $i$, $i = 1,\ldots,K-1$, and for each $j>i$ and each $\Jcal$ with
  $|\Jcal|=j$, given the state feedback of all the previous phases, the source searches for an index $\M{i}{\Jcal}$ such that
  $(\vecyhat{i}{\Jcal}(\M{i}{\Jcal}), \{\vecv{l}{\Ical}\}_{l<i, \Ical \subset \Jcal}, \vecs{i},
  \vecq{i})\in\Tcal^{n_i}_{\epsilon_{n_i}}(\Yhat{i}{\Jcal},\{V_{\Ical}\}_{ \Ical \subset  \Jcal}, S, Q)$\footnote{It is worth clarifying that the \emph{weak} typicality compatible with both discrete and continuous RVs is used. As in \cite{cover2012elements}, weak typicality is defined as $\Tcal^n_{\epsilon}(x^n)\defeq \{x^n:\vert -\frac{1}{n}\log p(x^n)- H(X) \vert \leq \epsilon\}$ for discrete RVs and $\Tcal^n_{\epsilon}(x^n)\defeq \{x^n:\vert -\frac{1}{n}\log p(x^n)- h(X) \vert \leq \epsilon\}$ for continuous RVs, where with a bit abuse of notation we denote pmf and probability density function with the same notation $p(x^n)$ for discrete and continuous cases, respectively.}. According to the covering
  lemma~\cite{elgamal_kim}, this is feasible with probability going to
  $1$ when $n_i\to\infty$, if 
  \begin{align}
   n_i \R{i}{\Jcal} \geq n_i I(\Yhat{i}{\Jcal};\{V_{\Ical}\}_{ \Ical \subset
   \Jcal}\cond \S{i},\Q{i}) + n_i\epsilon_{n_i}. \label{eq:source}
  \end{align}%
\end{enumerate}
\subsubsection*{\underline{Decoding}}
We focus on the decoding procedure of a particular receiver~$k\in\Kcal$ without loss of generality. At the end of phase $K$, a
backward decoding is performed. Specifically, for phase $j=K,K-1,\ldots,1$, the set of messages, 
 $\{M_{\Jcal}\}_{\Jcal\ni k}$, intended for user~$k$ is decode as follows.  
\begin{enumerate}
  \item For phase $j$, $j=K,K-1,\ldots,2$, by construction the ``future'' message set $\{\Mhat{j}{\Ucal}:\ {\Ucal\ni k},
    |\Ucal|>j\}$ has been decoded previously. The goal is to decode the ``current'' messages $\Mhat{i}{\Jcal}$ for each $\Jcal\ni
    k$ and each $i<j$. To that end, the decoder looks for a \emph{unique} index $\Mhat{i}{\Jcal}$ such that the following joint
    typicalities are satisfied \emph{simultaneously}
  \begin{align}
  \bigl(\vecv{i}{\Jcal}(\Mhat{i}{\Jcal}),
  \vecy{j}{k},\{\vecyhat{j}{\Ucal}(\Mhat{j}{\Ucal})\}_{\Ucal\supset \Jcal}, \vecs{j},
  \vecq{j} \bigr) &\in\Tcal^{n_j}_{\epsilon_{n_j}}(\V{i}{\Jcal}, Y_{k}, \{\Yhat{j}{\Ucal}\}_{
  \Ucal \supset \Jcal}, \S{j}, \Q{j}) \nonumber\\
  \bigl(\vecyhat{i}{\Jcal}(\Mhat{i}{\Jcal}),\vecy{i}{k},\vecs{i},\vecq{i}\bigr)
  &\in\Tcal^{n_i}_{\epsilon_{n_i}}(\Yhat{i}{\Jcal}, Y_{k}, \S{i}, \Q{i}).
\end{align}
   The probability that such an index cannot be
   found or is not correct~($\Mhat{i}{\Jcal} \ne \M{i}{\Jcal}$) vanishes when $n_i\to\infty$ provided that
  \begin{align}
    n_i \R{i}{\Jcal}  \leq   n_j I(\V{i}{\Jcal}; \Y{j}{k}\!\!, \{\Yhat{j}{\Ucal}\}_{
    \Ucal \supset \Jcal} \cond \S{j}\!\!\!, \Q{j}) + n_{i} I(\Yhat{i}{\Jcal};
    \Y{i}{k} \cond \S{i}\!\!\!, \Q{i}) - n_j\epsilon_{n_j} - n_i\epsilon_{n_i}.\label{ch-Kuser:eq:decodingcondition}
  \end{align}
  The error event analysis that leads to the above rate follows the exact same steps as the one in~\cite[Sec.~IV,
  p.1476]{tuncel2006slepian}, and is omitted here due to the space limitation.   
  \item Finally, for phase~$1$, the decoder searches for a unique $\hat{M}_{k}$ such that
  \begin{align}
    \bigl(\pmb{v}_{k}(\hat{M}_{k}), \vecy{1}{k},
    \{\vecyhat{1}{\Ucal}(\Mhat{1}{\Ucal})\}_{\Ucal\ni k},\vecs{1},\vecq{1}\bigr)  \in\Tcal^{n_1}_{\epsilon_{n_1}}(V_{k},
    Y_{k}, \{\Yhat{1}{\Ucal}\}_{ \Ucal \ni k}, \S{1}, \Q{1})\nonumber
    \end{align}
   According to the packing lemma, the probability that such an index cannot be
   found or is not correct~($ \hat{M}_k \ne M_k$) vanishes when $n_1\to\infty$ provided that
   \begin{align}
      n R_{k} \leq n_1 I(V_{k}; \Y{1}{k}, \{\Yhat{1}{\Ucal}\}_{ \Ucal \ni k} \cond
      \S{1}, \Q{1}) - n_1\epsilon_{n_1}. \label{ch-Kuser:eq:layeronecon} 
  \end{align}
\end{enumerate}
To summarize, $(R_1,\ldots,R_K)$ is achievable if for each $k$, \eqref{ch-Kuser:eq:layeronecon} is satisfied subject to
the existence of $\left\{ \R{i}{\Jcal}:\ i<j, |\Jcal|=j \right\}$ that verify \eqref{eq:source} and
\eqref{ch-Kuser:eq:decodingcondition}. Thus let $n_1,\ldots,n_K$ go to infinity by keeping the same ratio
$\alpha_1,\ldots,\alpha_n$, we obtain the rate \eqref{ch-Kuser:eq:IBofGBC1} if, for each $(k,i,j,\Jcal)$ with $i<j$ and $|\Jcal|=j$, 
\begin{align}
  I(\Yhat{i}{\Jcal};\{V_{\Ical}\}_{ \Ical \subset \Jcal}\cond \S{i},\Q{i}) \le \frac{\alpha_j}{\alpha_i}I(\V{i}{\Jcal}; \Y{j}{k}\!\!, \{\Yhat{j}{\Ucal}\}_{
    \Ucal \supset \Jcal} \cond \S{j}\!\!\!, \Q{j}) + I(\Yhat{i}{\Jcal};
    \Y{i}{k} \cond \S{i}\!\!\!, \Q{i}). \nonumber
\end{align}%
Using the Markovity $\Yhat{i}{\Jcal} \leftrightarrow (\{V_{\Ical}\}_{\Ical \subset \Jcal}, \S{i},\Q{i}) \leftrightarrow  Y_k^{(i)}$
and the chain rule, we have 
\begin{align}
  I(\Yhat{i}{\Jcal};\{V_{\Ical}\}_{ \Ical \subset \Jcal}\cond \S{i}\!\!\!,\Q{i}) - I(\Yhat{i}{\Jcal}; \Y{i}{k} \cond \S{i}\!\!\!, \Q{i}) &= I(\Yhat{i}{\Jcal};\{V_{\Ical}\}_{ \Ical \subset \Jcal}\cond \Y{i}{k}\!\!\!,\S{i}\!\!\!,\Q{i})
\end{align}%
which leads to the constraint \eqref{ch-Kuser:eq:IBofGBC1}. This completes of proof of Theorem~\ref{ch-Kuser:IBGBC}.

\section{Applications to the GBC and EBC}
\label{sec:application}

In this section, we apply the general result in Theorem~\ref{ch-Kuser:IBGBC} to the fading Gaussian BC and the
erasure BC. The key
is to fix the distribution \eqref{ch-Kuser:eq:pmf1} of the RVs involved in the rate region appropriately for each channel. 

In particular, the coded time-sharing random variable $\Q{j}$, $j\in\Kcal$, is used to indicate which of the messages $\M{i}{\Jcal}$
is to be sent.\footnote{Slightly abusing the subscript notation, we sometimes write $(\cdot)_k$ as $(\cdot)_{0\to\{k\}}$, e.g.,
the original message $M_k$ is also $\M{0}{\{k\}}$ }
Thus, it is natural to define $\Q{j}$ as $(\Q{j}_1, \Q{j}_2)$, where $\Q{j}_1\in \Qcal^{(j)}_1$ and $\Q{j}_2\in \Qcal^{(j)}_2$ 
with
\begin{align}
  \Qcal^{(j)}_1 &\defeq \begin{cases} 
    \{0\}, & j=1, \\
    \left\{ 1,\ldots,j-1 \right\}, & j>1,
  \end{cases}
  \quad \text{and} \quad
\Qcal^{(j)}_2 \defeq\left\{ \Jcal:\ |\Jcal|=j \right\}.
\end{align}%
We also let the channel input $X$ be a deterministic function of $V$ and $Q$. Specifically, when $\Q{j} = (i, \Jcal)$, 
the message $\M{i}{\Jcal}$ is carried by $\V{i}{\Jcal}$. Hence, we set
\begin{align}
X^{(j)} = \V{\Q{j}_1}{\Q{j}_2}. 
\end{align}%


\subsection{Fading Gaussian BC}
\label{ch-Kuser:sec:RsymN}

For the fading GBC, we focus on the symmetric channel and the corresponding symmetric rate for simplicity.  
To that end, we make the following choices on the RVs:
\begin{itemize}
  \item Time-sharing RVs. We let $\Q{j}_1$ be
    deterministic, namely, $\Q{j}_1 = j-1$, and let $\Q{j}_2$ be uniformly distributed over $\Qcal^{(j)}_2$, namely  
    \begin{align}
      \Pr\left( \Q{j}_2 = \Jcal \right) = \binom{K}{j}^{-1}, \quad \forall\, \Jcal \in \Qcal^{(j)}_2. 
      \label{ch-Kuser:eq:TS3}
    \end{align}%
    Intuitively, the above choice means that in phase $j$, we only transmit side information created in the previous phase $j-1$.
    This is similar to the general MAT scheme~\cite{maddah2012completely}. We can write $\Q{j} = (j-1,\Q{j}_2)$.
    The uniformity is simply due to the symmetry of the setting. 
  \item Gaussian distributed $V$'s. The RVs $V$'s with different subscripts are independent and identically distributed~(i.i.d.)~according
    to~$\mathcal{CN}(\pmb{0},\frac{P}{\nt}\IM_{\nt})$. It means that all the transmit antennas are used in each phase with
    isotropic signaling, which can be justified by the lack of instantaneous CSIT. 
  \item Side information $\hat{Y}$ as compression of the overheard signal. In phase~$i$, when $\Q{i} = (i-1, \Ical)$ for some
    $\Ical$, we set 
    \begin{align}
      \rvVecYhat{i}{\Jcal} &= \begin{cases} \rvMat{H}_{\Jcal\setminus\Ical} \rvVecX{i} +
        \hat{\rvVec{Z}}_{\Jcal\setminus\Ical}, & \text{if } \Jcal\supset\Ical \text{ and } |\Jcal|=i+1, \\
        0, & \text{otherwise}
      \end{cases}\label{eq:Yhat-GBC}
    \end{align}%
    where $\hat{\rvVec{Z}}_{\Jcal\setminus\Ical}\sim\mathcal{CN}(\pmb{0},\beta_i\sigma^2\IM)$ is the compression noise with
    $\beta_i>0$ being a parameter to be fixed later. The intuition behind \eqref{eq:Yhat-GBC} is the following. When $\Q{i} = (i-1,
    \Ical)$, the information intended for the users in the set $\Ical$ is being sent and is overheard by some user $k\not\in\Ical$.
    Let $\Jcal = \{k\}\cup\Ical$ be the new group. Then, the overheard signal $\rvMat{H}_{\Jcal\setminus\Ical} \rvVecX{i}$ is
    indeed interested by the users in group $\Ical$ since it provides an extra observation\footnote{When
    $\nr{1}+\cdots+\nr{K} \le \nt$, such
    observation is linearly independent of what each user in $\Ical$ already has.}. Furthermore, thanks to the joint source-channel coding, $\rvMat{H}_{\Jcal\setminus\Ical}
    \rvVecX{i}$ as a side information does not cost receiver~$k$ much to decode since it already has some noisy
    version of the information. 
\end{itemize}
Applying the above RVs to the general region in Theorem~\ref{ch-Kuser:IBGBC}, we obtain the following corollary. Some intermediate
steps are rather technical and deferred to Appendix~\ref{ch-Kuser:appendix:Gaussianregion}. 
\begin{corollary}\label{CH-KUSER:COL:GAU}
  For the $K$-user symmetric $\nt\times n_{\text{r}}$ fading GBC, the symmetric rate:
\begin{align}
  R_{\text{sym}}^{\text{GBC}} &= \max_{(\beta_1,\ldots,\beta_{K-1})\in\mathbb{R}_+^{K-1}} \left(K+ \sum_{j=2}^K \binom{K}{j} \prod_{t=2}^j
  \frac{\sum_{l\le t} b_{l,t} }{ a_{t}} \right)^{-1} a_1, \label{ch-Kuser:eq:RsymN}
\end{align}%
is achievable, where, for $t=1,\ldots,K$, 
\begin{align}
  a_{t} &\defeq \mathbb{E} \left[ \log\det\left( \IM + \snr \, \rvMat{H}_{\Tcal}^H
  \pmb{\Lambda}_t \rvMat{H}_{\Tcal} \right) \right] , \label{ch-Kuser:eq:a_t}\\
  b_{l,t} &\defeq \mathbb{E} \left[ \log\det\left( \IM + \frac{\snr}{\beta_{t-1}} \rvMat{H}_l
  (\IM +  {\snr}\, \rvMat{H}_1^H \rvMat{H}_1 )^{-1} \rvMat{H}_l^H
  \right) \right], \label{ch-Kuser:eq:b_t}
\end{align}%
with $\Tcal\defeq \{1\}\cup \{ t+1,\ldots,K \}$ and $\pmb{\Lambda}_t \defeq
\mathrm{diag}\bigl\{\IM_{n_{\text{r}}}, \beta_{t}^{-1} \IM_{(K-t)n_{\text{r}}}
\bigr\}$.
\end{corollary}
Although the maximization in \eqref{ch-Kuser:eq:RsymN} is not convex in general, it can be done numerically. We shall comment
more on this aspect in the next section with some examples.  

Now let us take a look at the high SNR regime. We consider the MISO case with $\nt=K$. We shall show from the above rate \eqref{ch-Kuser:eq:RsymN} that the optimal symmetric DoF can be achieved. To that
end, we let the compression noise variance be $\beta_i=1$, $i=1,\ldots,K-1$. 
From \eqref{ch-Kuser:eq:a_t} and \eqref{ch-Kuser:eq:b_t}, one can verify that, at high SNR,
\begin{align}
  a_t &= |\Tcal| \log \snr + O(1) = (K-t+1) \log\snr + O(1), \\
  b_{l,t} &= \begin{cases} O(1), & l=1, \\
    \log \snr + O(1), & l\ne1.
  \end{cases}
\end{align}%
Since the DoF is defined as $\textsf{DoF}_{\text{sym}}\defeq \lim_{\snr\to\infty} \frac{R_{\text{sym}}}{\log\snr}$, it follows
from \eqref{ch-Kuser:eq:RsymN} that
\begin{align}
  \textsf{DoF}_{\text{sym}}&= \left(K+ \sum_{j=2}^K
  \binom{K}{j} \prod_{t=2}^j \frac{t-1}{K-t+1} \right)^{-1} K \\
   &=\left(K+ \sum_{j=2}^K \binom{K}{j} \binom{K-1}{j-1}^{-1}
  \right)^{-1} K \\ 
   &=\biggl(\sum_{j=1}^K \frac{1}{j} \biggr)^{-1}, \label{CH-KUSER:EQ:DOF}
\end{align}%
which coincides with the optimal symmetric DoF derived in~\cite{maddah2012completely} for the same channel. Note that the
DoF achievability holds for all $\{\beta_i>0\}_i$ that do not scale with the SNR, while at finite SNR the exact values of the
$\beta$'s actually matter for the rate performance.

\subsection{Erasure BC}
\label{ch-Kuser:sec:EBC}
Next let us consider the erasure BC. We make the following choices on the RVs:
\begin{itemize}
  \item Time-sharing RVs. Let us recall that $\Q{j} = (\Q{j}_1, \Q{j}_2)$. Here we let $\Q{j}_1$ and $\Q{j}_2$ be independent for
    each $j\in\Kcal$, i.e., 
    \begin{align}
      p(\q{j}) &= p(\q{j}_1)p(\q{j}_2), \quad \forall\,\q{j} \in \Qcal^{(j)}. 
      \label{ch-Kuser:eq:TS1}
    \end{align}%
    However, we do not specify the distribution of $(\Q{j}_1, \Q{j}_2)$. 
  \item Uniformly distributed $V$'s. The RVs $V$'s with different subscripts are i.i.d.~over the input alphabet~$\Xcal$ according to a uniform distribution. Specifically,
    the distribution of $\V{i}{\Jcal}$, for each $i<j$ and $|\Jcal|=j$, is
    \begin{align}
      p(\v{i}{\Jcal}) &= \frac{1}{|\Xcal|}, \quad \forall\,\v{i}{\Jcal} \in \Xcal. 
    \end{align}%
    This choice guarantees the maximum entropy of the $V$'s, with $H(\V{i}{\Jcal})= \log|\Xcal|$.
  \item Side information $\hat{Y}$ as the overheard signal. In phase~$i$, when $\Q{i} = (l, \Ical)$ for some $l<i$ and $|\Ical| =
    i$, we set 
    \begin{align}
      \rvVecYhat{i}{\Jcal} &= \begin{cases} \X{i}, & \text{if } \Jcal\supset\Ical,\ S_{\Ical} \ne 
        \pmb{1},\ S_{\Jcal\setminus\Ical} = \pmb{1}, \text{and }S_{\bar{\Jcal}} = \pmb{0}, \\
        0, & \text{otherwise}.
      \end{cases}\label{eq:Yhat-EBC}
    \end{align}%
    The intuition behind \eqref{eq:Yhat-EBC} is the following. When $\Q{i} = (i-1,
    \Ical)$, the information intended for the users in the set $\Ical$ is being sent. 
    If this information is not received by some of the users in $\Ical$~(i.e. $S_{\Ical} \ne \pmb{1}$), and meanwhile received by
    some unintended users defined by $\Ucal$ with $\Ucal\cap\Ical=\emptyset$, then we define a new group 
    $\Jcal = \Ucal\cup\Ical$. We have the conditions $S_{\Jcal\setminus\Ical} = \pmb{1}$ and $S_{\bar{\Jcal}} =
    \pmb{0}$. 
    Thanks to the joint source-channel coding, such signal does not cost receivers in $\Jcal\setminus\Ical$ anything to decode
    since they already have the information. 
\end{itemize}

Applying the above RVs to the general region in Theorem~\ref{ch-Kuser:IBGBC}, we obtain the following corollary.
As in the Gaussian case, the technical intermediate
steps are deferred to Appendix~\ref{ch-Kuser:appendix:EBC}.   

\begin{corollary}\label{CH-KUSER:COL:ERA}
The rate tuple $(R_1,\ldots,R_K)\in\mathbb{R}^K_+$ is achievable in the EBC with state feedback 
if 
\begin{align}
R_{k} &\leq \alpha_{1} \Pr(\Q{1}_2 = k)(1-\delta_{\Kcal}) \log|\Xcal|,\label{ch-Kuser:eq:citeforAppendixA}\\
0 \leq \min_{j,k,\Jcal:\atop k\in \Jcal} \Bigl\{ \alpha_j \Pr(\Q{j}_2=\Jcal)& (1-\delta_{\Kcal  \setminus \Jcal\cup\{k\}})  -
\sum_{i=1}^{j-1}\alpha_i\!\!\sum_{\Ical\subset\Jcal,\Ical\ni k}\Pr(\Q{i}_2=\Ical)\P_{\Kcal  \setminus \Jcal\cup\{k\},\Jcal \setminus \Ical} \Bigr\}, 
\end{align}
for some $K$-tuple $(\alpha_1,\ldots,\alpha_K)\in\mathbb{R}^K_+$ with $\sum_k \alpha_k = 1$, and some distribution of
$\{\Q{j}_2\}_{j\in\Kcal}$. 
\end{corollary}




For the symmetric EBC, the above region coincides with the capacity region, as will be shown in
Corollary~\ref{CH-KUSER:CORO:SYMSPIND}. 

\begin{definition}
An EBC is said to be symmetric if the erasure probability $\delta_{\Ucal}$ only depends on the
cardinality of the set $\Ucal$, that is, $\delta_{\Ucal}=\delta_{\Ucal'}$ if $\vert\Ucal\vert=\vert\Ucal'\vert$.
\end{definition}

\begin{corollary}\label{CH-KUSER:CORO:SYMSPIND}
  The JSC scheme achieves the following  
capacity region of the symmetric EBC.
\begin{align}
  \Ccal_{\text{sym}}^{\text{EBC}}= \bigcap_{\pmb{\pi}} \left\{ \begin{array}{c}
  (R_{\pi(1)},\ldots,R_{\pi(K)})\in\mathbb{R}^K_+:  \\[1ex]
   \sum_{k=1}^K  \dfrac{R_{\pi(k)}}{1-\delta_{\{\pi(1),\cdots,\pi(k)\}}}    \leq \log|\Xcal| \end{array}\right\}, \label{ch-Kuser:eq:Csym} 
\end{align}%
where the intersection is over all permutations $\pmb{\pi} \defeq (\pi(1),\cdots,\pi(K))$ of $(1,\ldots,K)$.
\end{corollary}
\begin{proof}
The converse can be found in \cite{wang2012capacity,gatzianas2013multiuser} with the standard outer bound techniques by creating an artificial degraded BC. A detailed proof on the achievability is provided in Appendix~\ref{ch-Kuser:appendix:EBCsym}.
\end{proof}

\begin{remark}
The capacity region for the general EBC with state feedback is still unknown. In
\cite{wang2012capacity,gatzianas2013multiuser}, the authors designed a special scheme that can achieve the capacity region for
the general EBC with three users. In their capacity-achieving scheme, the transmitted signal can depend simultaneously on
messages from different phases, e.g., $M_{1}$ and $\M{2}{\{2,3\}}$. Such a result suggests that coded time-sharing may not be
enough to achieve the capacity with JSC scheme in general. We believe that it is possible to set the RVs in our region in a
similar way as the scheme in \cite{wang2012capacity,gatzianas2013multiuser} to achieve the three-user capacity region. However, it is out
of the scope of the current paper and is not considered here. Nevertheless, the capacity region beyond three users remains
unknown. 


\end{remark}
\section{Numerical Examples}\label{ch-Kuser:sec:Simulation}


In this section, we consider the Gaussian MISO channel with i.i.d.~Rayleigh
fading for $K=2$ and $3$ users. We let $\nt=K$ and evaluate the symmetric rate \eqref{ch-Kuser:eq:RsymN} of the
JSC scheme. The
maximization \eqref{ch-Kuser:eq:RsymN} over $\left\{ \beta_1,\ldots,\beta_{K-1} \right\}$ is done numerically. Since $K$ is
small in our examples, we simply sample each $\beta_i$ uniformly within a given region
of~$\beta_i$ with a small step size and then find out the maximum value of \eqref{ch-Kuser:eq:RsymN}.\footnote{Although the numerical maximum value with such a method may not be optimal, but it still represents an achievable
rate. }
For larger values of $K$, however, more sophisticated numerical methods may be needed.

Our scheme is compared to the following baseline schemes:
\begin{enumerate} 
\item The TDMA scheme. It is optimal for the no CSIT case, and achieves the following symmetric rate
  \begin{align}
    R_{\text{sym}}^{\text{TDMA}} &\defeq \frac{1}{K} I(X; Y_1 \cond S)\\
    &= \frac{1}{K} \mathbb{E} \left[ \log (1+ \snr \Vert\rvVec{H}_{1}\Vert^2) \right]. 
  \end{align}%
\item The original MAT scheme from \cite{maddah2012completely}. Here we consider the MAT scheme as described in
  \cite{maddah2012completely}, except for adding a proper linear scaling factor to meet the transmit power constraint. 
\item The generalized MAT~(GMAT) scheme from \cite{yi2013precoding}. 
  With GMAT, a precoder is designed to balance the alignment of interference and the enhancement of each signal. 
  The GMAT scheme includes the MAT scheme as a special case by
   letting the precoder be the respective channel matrix $\HM$ to reconstruct the overheard observations.
\item The quantized MAT~(QMAT) scheme from \cite{ali2014approximate}. Instead of sending the analogy linear
  combinations as the MAT scheme does, the QMAT transmits a \emph{quantized} version of each linear combination. In phase~$j$, 
  it turns out that $(j-1)(j+2)$ is the minimum quantization noise variance such that the message for group $\Jcal$ can be
  recovered at each user through a $(K-j+1) \times 1$ MISO channel. 
  The achievable symmetric rate is 
\begin{align}
  R_{\text{sym}}^{\text{QMAT}} = \frac{1}{K} {{\normalfont\textsf{DoF}_{\text{sym}}}}\, \mathbb{E} \left[ \log \det
  \bigl(\IM+ \snr\, \rvMat{H}_{\Kcal}\rvMat{H}_{\Kcal}^H \tilde{\NM}^{-1} \bigr) \right],  \label{ch-Kuser:eq:QMAT}
\end{align}
where $\tilde{\NM} \defeq \mathrm{diag} \Bigl(\{1+(j-1)(j+2)\}_{j=1:K}\Bigr)$ and {{\normalfont$\textsf{DoF}_{\text{sym}}$}} is
given by \eqref{CH-KUSER:EQ:DOF}. The fundamental differences between our JSC scheme and the QMAT are: 1) we use joint source-channel
coding while QMAT uses separate coding, 2) our source codebook is generated by $\hat{Y}$ that indicates what the users need while
QMAT explicitly generates linear combinations and the quantization of each combinations, and 3) the JSC scheme uses
all the transmit antennas all the time while the QMAT uses only a subset of $K-j+1$ transmit antennas in each phase~$j\in\Kcal$.
\item The genie-aided upper bound. 
For $k\in\Kcal$, a genie provides the output $Y_{k}$ to users $l\in\{k+1,\cdots,K\}$. 
The new channel can only have a larger capacity region than the original one, and it is a \emph{physically degraded} BC whose
capacity region cannot be enlarged with feedback. The single-letter characterization of the capacity region of such degraded BC is, for
some pmf $p(x\cond u_{K-1}) \prod_{k=2}^{K-1} p(u_k \cond u_{k-1}) p(u_1)$, is 
\begin{align}
R_k \le I(U_k;Y_1^k \cond U_{k-1}, S), \forall k\in\Kcal, \label{ch-Kuser:eq:UB}
\end{align}
where we define $U_K = X$ and $U_0 = 0$ for convenience~\cite{elgamal_kim}. 
Thus, the symmetric capacity of the original channel must satisfy \eqref{ch-Kuser:eq:UB}, which yields the following upper bound on the weighted sum
\begin{align}
  \sum_{k=1}^K \frac{C_{\text{sym}}}{k} & \le I(U_1;Y_1\cond S)+ \cdots + \frac{1}{K} I(U_K;Y_1^K\cond U_{K-1},S)\\
  & \le h(Y_1\cond S) + \sum_{k=1}^{K-1} \left(\frac{h(Y_1^{k+1}\cond U_k,S)}{k+1} -\frac{h(Y_1^k\cond U_k, S)}{k} \right) - \frac{1}{K} h(Y_1^K \cond X, S)\\
  & \le h(Y_1\cond S) - \frac{1}{K} h(Y_1^K \cond X, S)\\
  & = I(X;Y_1 \cond S) =  K R_{\text{sym}}^{\text{TDMA}},
\end{align}
where the second inequality is from the Markovity $h(Y_1^k\cond U_k, U_{k-1}, S) = h(Y_1^k\cond U_k, S)$ by the
construction of the pmf; 
the third inequality follows from the symmetry of the channel output in a symmetric fading
channel~\cite{vaze2012dofnoCSIT}, i.e., when $|\Jcal| \ge |\Ical|$,
 $\frac{h(Y_{\Jcal}\cond U,S)}{|\Jcal|} \le \frac{h(Y_{\Ical} \cond U, S)}{|\Ical|}$;
and the last equality holds since $h(Y_1^K \cond X, S) = h(Z_1^K \cond S) = K\, h(Y_1 \cond X, S)$. 
Hence, we have the following upper bound on the symmetric capacity
\begin{align}
  C_\text{sym} \le \textsf{DoF}_{\text{sym}} R_{\text{sym}}^{\text{TDMA}}.
\end{align}%
\end{enumerate}

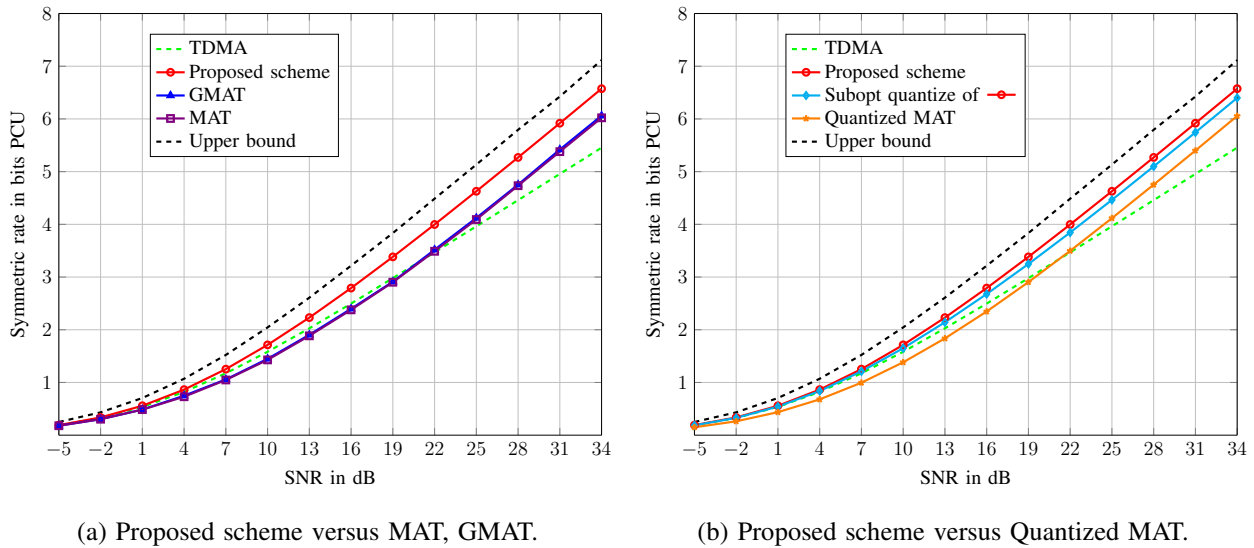
\begin{figure*}[t!]
    \centering
    \begin{subfigure}[t]{0.5\textwidth}
%
%
\definecolor{mycolor1}{rgb}{1.00000,0.00000,1.00000}%
%
{\centering
\resizebox {\textwidth} {!} {
\begin{tikzpicture}

\begin{axis}[%
width=4.521in,
height=3.515in,
at={(0.758in,0.474in)},
scale only axis,
separate axis lines,
every outer x axis line/.append style={black},
every x tick label/.append style={font=\color{black}},
xmin=-5,
xmax=34,
xtick={-5, -2,  1,  4,  7, 10, 13, 16, 19, 22, 25, 28, 31, 34},
xlabel={SNR in dB},
xmajorgrids,
every outer y axis line/.append style={black},
every y tick label/.append style={font=\color{black}},
ymin=0,
ymax=8,
ytick={ 1,  2,  3,  4,  5,  6,  7,  8,  9, 10, 11, 12, 13, 14, 15, 16, 17, 18},
ylabel={Symmetric rate in bits PCU},
ymajorgrids,
axis background/.style={fill=white},
legend style={at={(0.168,0.66)},anchor=south west,legend cell align=left,align=left,draw=black}
]

\addplot [color=green,dashed,line width=1.3pt]
  table[row sep=crcr]{%
-5	0.187546540538113\\
-2	0.327237282678585\\
1	0.536288359584215\\
4	0.820538908673256\\
7	1.17389984958115\\
10	1.58176493812021\\
13	2.02733681093898\\
16	2.49640153377683\\
19	2.97901396497745\\
22	3.46904859598788\\
25	3.96303430952266\\
28	4.45909655366292\\
31	4.95624274607616\\
34	5.45395005691729\\
};
\addlegendentry{TDMA};

\addplot [color=red,solid,line width=1.3pt,mark=o]
  table[row sep=crcr]{%
-5	0.190347371796437\\
-2	0.336087778369501\\
1	0.557626564084257\\
4	0.864062026136201\\
7	1.25288222638402\\
10	1.71310540926666\\
13	2.23045187139519\\
16	2.79085664116241\\
19	3.38276146333855\\
22	3.99726212539795\\
25	4.62775673409968\\
28	5.26943397913751\\
31	5.91888858318266\\
34	6.57402286841659\\
};
\addlegendentry{Proposed scheme};

\addplot [color=blue,solid,line width=1.3pt,mark=triangle]
  table[row sep=crcr]{%
-5	0.179784332445417\\
-2	0.302998854929873\\
1	0.48555130548121\\
4	0.748626416013315\\
7	1.05727362705386\\
10	1.44799010091917\\
13	1.90417628536524\\
16	2.39790595477219\\
19	2.91150643010322\\
22	3.51869539097274\\
25	4.12163796146211\\
28	4.75357880019832\\
31	5.41902340600179\\
34	6.06091152101611\\
};
\addlegendentry{GMAT};

\addplot [color=violet,solid,line width=1.3pt,mark=square]
  table[row sep=crcr]{%
-5	0.178788426749417\\
-2	0.302885739324498\\
1	0.48451270528121\\
4	0.728862550413315\\
7	1.04727363501686\\
10	1.42799010091917\\
13	1.88453652176284\\
16	2.37754772190599\\
19	2.90150639200973\\
22	3.48869430152274\\
25	4.09163141028311\\
28	4.73357967886192\\
31	5.37911040526079\\
34	6.02091110100236\\
};
\addlegendentry{MAT};

\addplot [color=black,dashed,line width=1.3pt]
  table[row sep=crcr]{%
-5	0.249966397778481\\
-2	0.432895476925853\\
1	0.704897382891711\\
4	1.07024831199751\\
7	1.52218232697837\\
10	2.04417097892077\\
13	2.60848211176609\\
16	3.21565067599771\\
19	3.83409996150396\\
22	4.48699940333282\\
25	5.13358391812318\\
28	5.79379794106812\\
31	6.42459832530775\\
34	7.1162532118365\\
};
\addlegendentry{Upper bound};

\end{axis}
\end{tikzpicture}%
}
\caption{Proposed scheme versus MAT, GMAT.}
\label{ch-Kuser:fig:twoUMATProposed}
}


    \end{subfigure}%
    ~ 
    \begin{subfigure}[t]{0.5\textwidth}
%
%
\definecolor{mycolor1}{rgb}{1.00000,0.00000,1.00000}%
%
{\centering
\resizebox {\textwidth} {!} {
\begin{tikzpicture}

\begin{axis}[%
width=4.521in,
height=3.515in,
at={(0.758in,0.474in)},
scale only axis,
separate axis lines,
every outer x axis line/.append style={black},
every x tick label/.append style={font=\color{black}},
xmin=-5,
xmax=34,
xtick={-5, -2,  1,  4,  7, 10, 13, 16, 19, 22, 25, 28, 31, 34},
xlabel={SNR in dB},
xmajorgrids,
every outer y axis line/.append style={black},
every y tick label/.append style={font=\color{black}},
ymin=0,
ymax=8,
ytick={ 1,  2,  3,  4,  5,  6,  7,  8,  9, 10, 11, 12, 13, 14, 15, 16, 17, 18},
ylabel={Symmetric rate in bits PCU},
ymajorgrids,
axis background/.style={fill=white},
legend style={at={(0.168,0.66)},anchor=south west,legend cell align=left,align=left,draw=black}
]

\addplot [color=green,dashed,line width=1.3pt]
  table[row sep=crcr]{%
-5	0.187546540538113\\
-2	0.327237282678585\\
1	0.536288359584215\\
4	0.820538908673256\\
7	1.17389984958115\\
10	1.58176493812021\\
13	2.02733681093898\\
16	2.49640153377683\\
19	2.97901396497745\\
22	3.46904859598788\\
25	3.96303430952266\\
28	4.45909655366292\\
31	4.95624274607616\\
34	5.45395005691729\\
};
\addlegendentry{TDMA};

\addplot [color=red,solid,line width=1.3pt,mark=o]
  table[row sep=crcr]{%
-5	0.190347371796437\\
-2	0.336087778369501\\
1	0.557626564084257\\
4	0.864062026136201\\
7	1.25288222638402\\
10	1.71310540926666\\
13	2.23045187139519\\
16	2.79085664116241\\
19	3.38276146333855\\
22	3.99726212539795\\
25	4.62775673409968\\
28	5.26943397913751\\
31	5.91888858318266\\
34	6.57402286841659\\
};
\addlegendentry{Proposed scheme};\label{ch-Kuser:curve:twoProp}

\addplot [color=cyan,solid,line width=1.3pt,mark=diamond]
  table[row sep=crcr]{%
-5	0.186570657225929\\
-2	0.328134545841131\\
1	0.542499579831951\\
4	0.837698422688542\\
7	1.21036490574486\\
10	1.64942588545177\\
13	2.14214651833402\\
16	2.67776409872826\\
19	3.24784591186231\\
22	3.84534557207308\\
25	4.4640615220525\\
28	5.09867848402164\\
31	5.74488068572664\\
34	6.39928657583589\\
};
\addlegendentry{Subopt quantize of \ref{ch-Kuser:curve:twoProp}};

\addplot [color=orange,solid,line width=1.3pt,mark=star]
  table[row sep=crcr]{%
-5	0.14950147704746\\
-2	0.26297177613729\\
1	0.436262702238935\\
4	0.679353844096745\\
7	0.995594465386455\\
10	1.38302560878906\\
13	1.83615509728285\\
16	2.34670141583657\\
19	2.90425010579863\\
22	3.49764860930324\\
25	4.11666875038826\\
28	4.75306558361748\\
31	5.40076395355474\\
34	6.05552537719563\\
};
\addlegendentry{Quantized MAT};


\addplot [color=black,dashed,line width=1.3pt]
  table[row sep=crcr]{%
-5	0.249966397778481\\
-2	0.432895476925853\\
1	0.704897382891711\\
4	1.07024831199751\\
7	1.52218232697837\\
10	2.04417097892077\\
13	2.60848211176609\\
16	3.21565067599771\\
19	3.83409996150396\\
22	4.48699940333282\\
25	5.13358391812318\\
28	5.79379794106812\\
31	6.42459832530775\\
34	7.1162532118365\\
};
\addlegendentry{Upper bound};

\end{axis}
\end{tikzpicture}%
}
\caption{Proposed scheme versus Quantized MAT.}
\label{ch-Kuser:fig:twoUvsQuantize}
}

    \end{subfigure}
    \caption{The proposed scheme versus the baseline schemes: two-user BC.}
    \label{fig:2user}
\end{figure*}


The two-user and three-user cases are evaluated separately in Fig.~\ref{fig:2user} and Fig.~\ref{fig:3user}, respectively. 
In both Fig.~\ref{ch-Kuser:fig:twoUvsQuantize} and Fig.~\ref{ch-Kuser:fig:threeUvsQuantize}, 
the curve \ref{ch-Kuser:curve:threePropsubopt} denotes a variant of our proposed scheme where the quantization noises $\beta_j$'s
are not optimized. Instead, we apply the same equivalent compression noise variance used in the QMAT scheme, that is,
$\beta_{j-1}=1+(j-1)(j+2)$. We have the following comments on the results. 
\begin{itemize}
\item 
  From the plots, we see that the curves of the GMAT scheme proposed in \cite{yi2013precoding} and the MAT curves almost
  overlap in all SNR regime for $K=2,3$. It shows that the performance improvement brought by carefully designing the linear
  combinations~(referred to as precoder) is marginal in the i.i.d.~isotropic fading case. Another generalization direction within
  the MAT framework is the quantization of linear combinations. Although the MAT and QMAT are not compared directly in the same
  plot, we can still observe that the QMAT does outperform the MAT scheme especially when $K=3$ in medium-to-high
  SNR regime. However, appreciable gain appears only at high SNR. 
\item In the low-to-medium SNR regime, the MAT/GMAT/QMAT schemes are outperformed by the TDMA. This result is somewhat surprising
  since, unlike the other schemes, TDMA does not exploit the state feedback. Indeed, the MAT-like schemes use the state feedback to
  perform interference alignment which is known to be optimal at high SNR but is usually less good when the SNR is not high. In
  such regime, the channel is not interference limited and sending linear equations may be too costly for the marginal interference
  mitigation effect. In the high SNR regime, the MAT-like schemes dominates the TDMA scheme eventually thanks to a larger DoF gain,
  which is reflected by the slopes of the curves. The optimal DoF of the MAT-like schemes is also confirmed by the fact that the
  corresponding curves are almost parallel to the upper bound curve.

\begin{figure*}[t!]
    \centering
    \begin{subfigure}[t]{0.5\textwidth}
%
%
\definecolor{mycolor1}{rgb}{1.00000,0.00000,1.00000}
%
{\centering
\resizebox {\textwidth} {!} {
\begin{tikzpicture}

\begin{axis}[%
width=4.521in,
height=3.515in,
at={(0.758in,0.474in)},
scale only axis,
separate axis lines,
every outer x axis line/.append style={black},
every x tick label/.append style={font=\color{black}},
xmin=-5,
xmax=34,
xtick={-5, -2,  1,  4,  7, 10, 13, 16, 19, 22, 25, 28, 31, 34},
xlabel={SNR(dB)},
xmajorgrids,
every outer y axis line/.append style={black},
every y tick label/.append style={font=\color{black}},
ymin=0,
ymax=7,
ytick={0.5,   1, 1.5,   2, 2.5,   3, 3.5,   4, 4.5,   5, 5.5,   6, 6.5,   7, 7.5,   8},
ylabel={Symmetric rate (bits per channel use)},
ymajorgrids,
axis background/.style={fill=white},
legend style={at={(0.165,0.629)},anchor=south west,legend cell align=left,align=left,draw=black}
]

\addplot [color=green,dashed,line width=1.3pt]
  table[row sep=crcr]{%
-5	0.301771884315411\\
-2	0.476837140736828\\
1	0.701586989294782\\
4	0.966675735763477\\
7	1.25999195934402\\
10	1.57088280796413\\
13	1.89183326886402\\
16	2.21822890709335\\
19	2.5474720602495\\
22	2.87817495060288\\
25	3.20961805344472\\
28	3.54143436259208\\
31	3.87343828364017\\
34	4.20553637646208\\
};
\addlegendentry{TDMA};\label{lineTDMA}

\addplot [color=red,solid,line width=1.3pt,mark=o]
  table[row sep=crcr]{%
-5	0.331645960792184\\
-2	0.541449455431748\\
1	0.823700984045328\\
4	1.17228498656287\\
7	1.57654213166155\\
10	2.02530356896868\\
13	2.50600310632889\\
16	3.0091543356911\\
19	3.52764037779381\\
22	4.05553272712782\\
25	4.58963576730721\\
28	5.1276320076776\\
31	5.66790964162738\\
34	6.20952611213445\\
};
\addlegendentry{Proposed scheme};

\addplot [color=blue,solid,line width=1.3pt,mark=triangle]
  table[row sep=crcr]{%
-5	0.26775085057974\\
-2	0.417304273561489\\
1	0.618315327711643\\
4	0.862157930691768\\
7	1.15246334434786\\
10	1.48632629290284\\
13	1.85602218329143\\
16	2.26937092292879\\
19	2.6925306659039\\
22	3.14291980927657\\
25	3.62546625018125\\
28	4.10573817958151\\
31	4.59618529604007\\
34	5.0988560122337\\
};
\addlegendentry{GMAT};
  
\addplot [color=violet,solid,line width=1.3pt,mark=square]
  table[row sep=crcr]{%
-5	0.257783967157747\\
-2	0.402039820751923\\
1	0.594086810281968\\
4	0.835406888207668\\
7	1.12446457214482\\
10	1.45752287829792\\
13	1.8296059359264\\
16	2.23529769102872\\
19	2.66911629327713\\
22	3.12557329234901\\
25	3.59917240216885\\
28	4.08449480107646\\
31	4.57642100494914\\
34	5.07047584615221\\
};
\addlegendentry{MAT};

\addplot [color=black,dashed,line width=1.3pt]
  table[row sep=crcr]{%
-5	0.493808537970673\\
-2	0.780278957569355\\
1	1.14805143702783\\
4	1.58183302215842\\
7	2.06180502438112\\
10	2.57053550394131\\
13	3.09572716723203\\
16	3.62982912069822\\
19	4.16859064404463\\
22	4.70974082825926\\
25	5.25210226927318\\
28	5.7950744115143\\
31	6.33835355504755\\
34	6.88178679784704\\
};
\addlegendentry{Upper bound};

\end{axis}
\end{tikzpicture}%
}
\caption{Proposed scheme versus MAT, GMAT.}
\label{ch-Kuser:fig:threeUMATProposed}
}

%
    \end{subfigure}%
    ~ 
    \begin{subfigure}[t]{0.5\textwidth}
%
%
\definecolor{mycolor1}{rgb}{0.00000,1.00000,1.00000}
%
{\centering
\resizebox {\textwidth} {!} {
\begin{tikzpicture}

\begin{axis}[%
width=4.456in,
height=3.481in,
at={(0.747in,0.47in)},
scale only axis,
separate axis lines,
every outer x axis line/.append style={black},
every x tick label/.append style={font=\color{black}},
xmin=-5,
xmax=34,
xtick={-5, -2,  1,  4,  7, 10, 13, 16, 19, 22, 25, 28, 31, 34},
xlabel={SNR in dB},
xmajorgrids,
every outer y axis line/.append style={black},
every y tick label/.append style={font=\color{black}},
ymin=0,
ymax=7,
ytick={0.5,   1, 1.5,   2, 2.5,   3, 3.5,   4, 4.5,   5, 5.5,   6, 6.5,   7, 7.5,   8},
ylabel={Symmetric rate in bits PCU},
ymajorgrids,
axis background/.style={fill=white},
legend style={at={(0.15,0.648)},anchor=south west,legend cell align=left,align=left,draw=black}
]

\addplot [color=green,dashed,line width=1.3pt]
  table[row sep=crcr]{%
-5	0.300127645422714\\
-2	0.474377544359828\\
1	0.698292072341991\\
4	0.962693775820847\\
7	1.25554969891168\\
10	1.5661777353608\\
13	1.88699198702405\\
16	2.21331984201564\\
19	2.54252953941912\\
22	2.87321586338697\\
25	3.20465072483643\\
28	3.53646292046476\\
31	3.86846478433185\\
34	4.20056184726784\\
};
\addlegendentry{TDMA};

\addplot [color=red,solid,line width=1.3pt,mark=o]
  table[row sep=crcr]{%
-5	0.332203643230616\\
-2	0.541523542376721\\
1	0.822618764077206\\
4	1.16931162891432\\
7	1.57059408479083\\
10	2.0161187830935\\
13	2.49343834928073\\
16	2.99353070938056\\
19	3.50975653141558\\
22	4.03587122116389\\
25	4.56831386896588\\
28	5.10519289464525\\
31	5.64482124476757\\
34	6.18599215082298\\
};
\addlegendentry{Proposed scheme};\label{ch-Kuser:curve:threeProp}

\addplot [color=cyan,solid,line width=1.3pt,mark=diamond]
  table[row sep=crcr]{%
-5	0.322926737288566\\
-2	0.52346973877581\\
1	0.792212890353638\\
4	1.1231561585912\\
7	1.50476222739505\\
10	1.92584590629375\\
13	2.37840868905399\\
16	2.85685214549232\\
19	3.35632600443112\\
22	3.87206139704364\\
25	4.39958086900216\\
28	4.93510258870414\\
31	5.47573349412329\\
34	6.0194125534816\\
};
\addlegendentry{Subopt quantize of \ref{ch-Kuser:curve:threeProp}};\label{ch-Kuser:curve:threePropsubopt}

\addplot [color=orange,solid,line width=1.3pt,mark=star]
  table[row sep=crcr]{%
-5	0.218079837750247\\
-2	0.354271362558238\\
1	0.542068790327693\\
4	0.784853112858585\\
7	1.08255545601698\\
10	1.43207623623917\\
13	1.82812537331855\\
16	2.26407571072935\\
19	2.7325291467204\\
22	3.2259132509524\\
25	3.73725214850026\\
28	4.26080147129932\\
31	4.79227484639119\\
34	5.32870915267647\\
};
\addlegendentry{Quantized MAT};

\addplot [color=black,dashed,line width=1.3pt]
  table[row sep=crcr]{%
-5	0.491117965237169\\
-2	0.7762541634979\\
1	1.14265975474144\\
4	1.57531708770684\\
7	2.05453587094638\\
10	2.56283629422677\\
13	3.08780506967571\\
16	3.62179610511651\\
19	4.16050288268583\\
22	4.70162595826958\\
25	5.24397391336871\\
28	5.78693932439687\\
31	6.33021510163394\\
34	6.87364665916556\\
};
\addlegendentry{Upper bound};

\end{axis}
\end{tikzpicture}%
}
\caption{Proposed scheme versus Quantized MAT.}
\label{ch-Kuser:fig:threeUvsQuantize}
}
    \end{subfigure}
    \caption{The proposed scheme versus the baseline schemes: three-user BC.}
    \label{fig:3user}
\end{figure*}
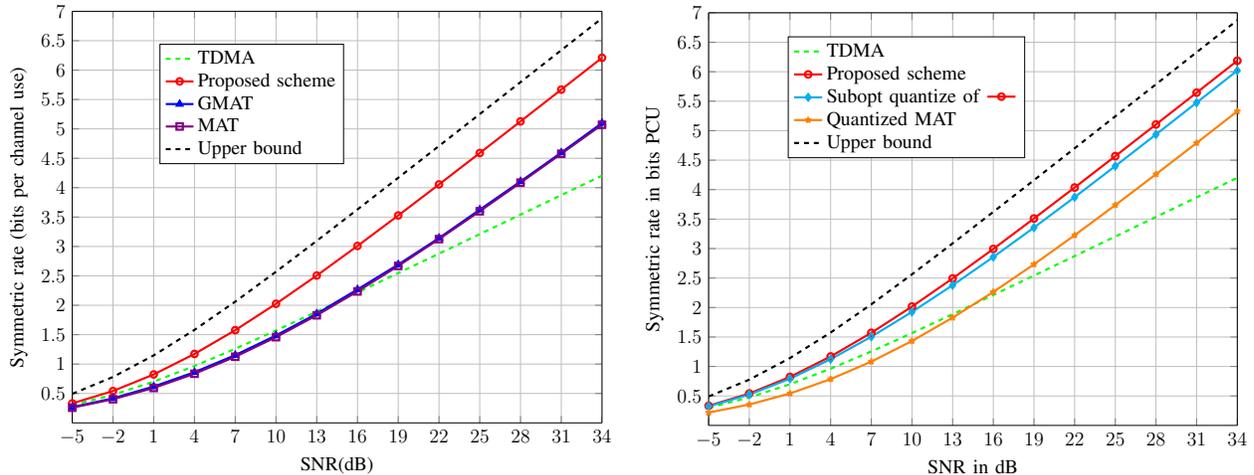

\item In all SNR {regimes}, the proposed JSC scheme outperforms all four baseline schemes~(TDMA, MAT, GMAT, QMAT) and has a non-negligible
  power gain over the MAT-like schemes. This gain becomes more appealing in the medium-to-low SNR regime in which
  the MAT-like schemes are not even better than the simple TDMA scheme. Our scheme can still take advantage of the state feedback to achieve a
  better performance. This is mainly thanks to the flexibility over the duration of each phase~(time-slot) and the compression parameters as a function
  of the SNR, which is not possible with the MAT/GMAT schemes. The comparison to the QMAT scheme is even more
  interesting, since both the JSC scheme and
  QMAT are based on compression. We see that the performance gain over QMAT is almost 3~dB for $K=2$ and is up to $6$~dB for $K=3$. To analyze the
  causes of such a significant gain, we fix the $\beta_{j-1}=1+(j-1)(j+2)$ which corresponds to the same setting in the QMAT.  
  As can be seen from the plots in Fig.~\ref{ch-Kuser:fig:twoUvsQuantize} and
  Fig.~\ref{ch-Kuser:fig:threeUvsQuantize}, the JSC scheme still dominates
  the QMAT with a slight performance degradation from the case with optimized $\beta$'s.
  Such an observation suggests that the main performance gain of our scheme over the QMAT comes from the joint source-channel coding.
\end{itemize}


\section{Conclusion}\label{ch-Kuser:sec:Conclusion}
In this paper, we proposed a novel scheme for the general state-dependent $K$-user broadcast channel with state
feedback. The proposed scheme is based on joint source-channel coding and coded time-sharing. Thanks to the
systematic and scalable structure of this scheme, we managed to derive the corresponding achievable region in
terms of a reasonable number of parameters. Such region was then evaluated for two special cases, namely,    the
erasure BC and fading Gaussian BC. In particular, we showed that our results covered the previously known
capacity region for the erasure BC. In addition, for the fading Gaussian BC, we demonstrated through numerical
evaulation a non-negligible power gain of our scheme over the existing ones in the literature. We argued that
such a substantial performance gain comes from the use of joint source-channel coding which is still highly
theoretical. 
Practical implementation of such schemes would be an interesting direction to explore in the future.

\appendices

\section{Proof of Corollary~\ref{CH-KUSER:COL:GAU}}\label{ch-Kuser:appendix:Gaussianregion}
To study the symmetric rate, we let  $R_1 = \cdots = R_K = R_{\text{sym}}$. Due to the symmetry of the channel, it
is without loss of generality to consider receiver~1.  
In particular, we apply the RVs choice given in Section~\ref{ch-Kuser:sec:RsymN}, and compute the quantities in
\eqref{ch-Kuser:eq:IBofGBC2} and \eqref{ch-Kuser:eq:IBofGBC1}. {We define $\tilde{\Jcal}\defeq
\{1\}\cup\{j+1,\cdots,K\}$, $\rvVec{W}_j\defeq \Bigl[ \rvVec{Z}_{1}^T,
\hat{\rvVec{Z}}_{\Ucal_{j+1}\setminus\Jcal}^T,\cdots, \hat{\rvVec{Z}}_{\Ucal_{K}\setminus\Jcal}^T \Bigr]^T$.} Then, we have
\begin{align}
  \MoveEqLeft{I(V_{1}; \Y{1}{1}, \{\Yhat{1}{\Ucal}\}_{\Ucal \ni k}\cond \S{1}, \Q{1})} \nonumber \\
  &= \Pr(\Q{1}=(0,\{1\})) {I(V_{1}; \Y{1}{1}, \Yhat{1}{\{1,2\}}, \ldots, \Yhat{1}{\{1,K\}} \cond \S{1}, \Q{1}=(0,\{1\}))} \\
    &= K^{-1} I(\rvVec{V}_{1}; \rvMat{H}_{\Kcal} \rvVec{V}_{1} + \rvVec{W}_1 \cond \rvMat{H}_{\Kcal}) \\ 
    &= K^{-1} \mathbb{E}\left[ \log\det\left( \IM + \snr \, \rvMat{H}_{\Kcal}^H \pmb{\Lambda}_1 \rvMat{H}_{\Kcal}
    \right) \right]
    = K^{-1} a_1.
  \end{align}
  Similarly, assuming $i=j-1$ and $|\Ucal| = |\Jcal|+1=j+1$, we obtain
\begin{align}
{I(\V{i}{\Jcal};\Y{j}{1},\{\Yhat{j}{\Ucal}\}_{ \Ucal  \supset \Jcal} \cond \S{j}, \Q{j})} 
&= \binom{K}{j}^{-1} I(\rvVec{V}_{i\to\Jcal}; \rvMat{H}_{\tilde{\Jcal}} \rvVec{V}_{i\to\Jcal} + \rvVec{W}_j \cond \rvMat{H}_{\tilde{\Jcal}}) \nonumber  \\
    &= \binom{K}{j}^{-1} \mathbb{E} \left[ \log\det\left( \IM + \snr \, \rvMat{H}_{\tilde{\Jcal}}^H
  \pmb{\Lambda}_j \rvMat{H}_{\tilde{\Jcal}} \right) \right] = \binom{K}{j}^{-1} a_j, \nonumber\\
{I(\{V_{\Ical}\}_{\Ical \subset \Jcal}; \Yhat{i}{\Jcal} \cond \Y{i}{1}, \S{i}, \Q{i})} 
&= \sum_{l=1}^j \binom{K}{i}^{-1} I(\rvVec{V}_{\Ical}; \rvMat{H}_{l} \rvVec{V}_{\Ical}+\hat{\rvVec{Z}}_{\Jcal
\setminus \Ical} \cond \rvMat{H}_{1} \rvVec{V}_{\Ical}+\rvVec{Z}_1, \rvMat{H}_{1},\rvMat{H}_{l}) \nonumber\\
    &= \sum_{l=1}^j \binom{K}{i}^{-1} \mathbb{E} \left[ \log\det\left( \IM + \frac{\snr}{\beta_i} \rvMat{H}_l
  (\IM +  {\snr}\, \rvMat{H}_1^H \rvMat{H}_1 )^{-1} \rvMat{H}_l^H
  \right) \right] \nonumber\\
  &= \sum_{l=1}^j \binom{K}{i}^{-1} b_{l,j}. \label{ch-Kuser:eq:Gausrea2}
  \end{align}%
  Thus, we can rewrite the Gaussian rate region as below.
\begin{align}
  R_{\text{sym}} &\leq \alpha_1 K^{-1} a_1 \label{ch-Kuser:eq:tmp291}\\
0 &\leq \alpha_j \binom{K}{j}^{-1} a_j  - \alpha_{j-1}\binom{K}{j-1}^{-1}\sum_{l=1}^{j} b_{l,j}. \label{ch-Kuser:eq:Gaussianalphaj}
\end{align}

For a given set of $\{\beta_j\}$ and for a given SNR, $\{a_j\}$ and $\{b_{l,j}\}$ are fixed. Hence, the maximum
achievable rate of $R_\text{sym}$ can be reached when $\alpha_1$ is maximized. However, the selection of
$\{\alpha_i\}$ is subject to the constraint $\sum_{j=1}^{K}\alpha_j=1$. Applying \eqref{ch-Kuser:eq:Gaussianalphaj}
$K-1$ times, for $j=K,K-1,\ldots,2$, we obtain
  \begin{align}
    \alpha_1 \le c_1 \alpha_2 \le \cdots \le c_{K-1} \alpha_K,
    \label{ch-Kuser:eq:alphamultiineq}
  \end{align}%
  where $c_k$'s are nonnegative and can be found from \eqref{ch-Kuser:eq:Gaussianalphaj}. We argue that from 
    \eqref{ch-Kuser:eq:alphamultiineq} it is without loss of optimality to assume that
    \eqref{ch-Kuser:eq:Gaussianalphaj} should hold with equality for all $j$. To see this, let
    $\{\alpha_j,j\in\Kcal\}$ be such that $\sum_{j=1}^K \alpha_j=1$ and assume that some of the inequalities in
    \eqref{ch-Kuser:eq:alphamultiineq} are strict. Then, we can always reduce some of $\alpha_2,\ldots,\alpha_K$ 
and make sure that all the equalities hold, which would in turn lower the value of summation of $\alpha$, i.e.,
$\sum_{j=1}^{K}\alpha_j=c<1$. In this case, we can make the scaling $\alpha^*_j=\frac{\alpha_j}{c}$ which increases
$\alpha_1$ and also the objective function. With this reasoning, we conclude that the optimal value of $\alpha_1$
should be such that \eqref{ch-Kuser:eq:Gaussianalphaj} holds with equality for all $j$, which leads to 
\begin{align}
\alpha^*_1=\left(1+\sum_{j=2}^K \frac{\binom{K}{j}}{K} \prod_{t=2}^j \frac{\sum_{l\le t} b_{l,t} }{
a_{t}}\right)^{-1}.
\end{align}%
Plugging $\alpha^*_1$ back to~\eqref{ch-Kuser:eq:tmp291}, we obtain the optimal symmetric rate
\eqref{ch-Kuser:eq:RsymN} in the Gaussian case.

\section{Proof of Corollary~\ref{CH-KUSER:COL:ERA}}\label{ch-Kuser:appendix:EBC}
In the following, we first apply the RVs selected in Section~\ref{ch-Kuser:sec:EBC} and evaluate the quantities
in~\eqref{ch-Kuser:eq:IBofGBC2} and \eqref{ch-Kuser:eq:IBofGBC1}, that is, for user $k\in\Kcal$,
\begin{align}
I(V_{k}; \Y{1}{k}, \{\Yhat{1}{\Ucal}\}_{\Ucal \ni k}  \cond \S{1}, \Q{1}) 
	&= \Pr(\Q{1}_2 = k)(1-\delta_{\Kcal}) \log|\Xcal|,  \label{ch-Kuser:eq:EBCspe1}\\ 
I(\V{i}{\Jcal};\Y{j}{k},\{\Yhat{j}{\Ucal}\}_{ \Ucal \supset \Jcal} \cond \S{j}, \Q{j})
&= \Pr(\Q{j}=(i,\Jcal)) (1-\delta_{\bar{\Jcal}\cup\{k\}})\log|\Xcal|, \label{ch-Kuser:eq:EBCspe2}\\
&= \Pr(\Q{j}_1=i) \Pr(\Q{j}_2=\Jcal) (1-\delta_{\bar{\Jcal}\cup\{k\}})\log|\Xcal| \\
I(\{V_{\Ical}\}_{\Ical \subset \Jcal}; \Yhat{i}{\Jcal}  \cond \Y{i}{k}, \S{i}, \Q{i})
&=\sum_{\Ical\subset\Jcal}\Pr(\Q{i}_2=\Ical)I(V_{\Ical}; \Yhat{i}{\Jcal}  \cond \Y{i}{k}, \S{i}, \Q{i}_1, \Q{i}_2 =
\Ical)\nonumber\\
&= \sum_{\Ical\subset\Jcal,\Ical\ni k}\Pr(\Q{i}_2=\Ical, \Yhat{i}{\Jcal} \neq 0, \Y{i}{k}=\,?)H(V_{\Ical})\label{ch-Kuser:eq:EBCspe3}\\
&= \sum_{\Ical\subset\Jcal,\Ical\ni k}\Pr(\Q{i}_2=\Ical)\P_{\bar{\Jcal}\cup\{k\},\Jcal \setminus \Ical}
\log|\Xcal|, \label{ch-Kuser:eq:EBCspe4}
\end{align}%
where \eqref{ch-Kuser:eq:EBCspe1} can be interpreted as: receiver~k can recover the intended signal on $M_k$ unless
all the receivers are in erasure; \eqref{ch-Kuser:eq:EBCspe2} and \eqref{ch-Kuser:eq:EBCspe4} are obtained with the
same reasoning on the choice of the side information $\Yhat{i}{\Jcal}$ as defined in \eqref{eq:Yhat-EBC}. From
\eqref{ch-Kuser:eq:IBofGBC2} and \eqref{ch-Kuser:eq:EBCspe1}, we obtain \eqref{ch-Kuser:eq:citeforAppendixA}. 


Applying \eqref{ch-Kuser:eq:IBofGBC1}, we have, for all $i,j,k,\Jcal$ with $i<j$ and $k\in\Ical$, 
\begin{align*}
0\leq \alpha_j \Pr(\Q{j}_1=i) \Pr(\Q{j}_2=\Jcal) (1-\delta_{\bar{\Jcal}\cup\{k\}})\log|\Xcal| - \alpha_i \sum_{\Ical\subset\Jcal,\Ical\ni k}\Pr(\Q{i}_2=\Ical)\P_{\bar{\Jcal}\cup\{k\},\Jcal \setminus \Ical}\log|\Xcal|.
\end{align*}
We assume that the probabilities and $\alpha$'s are bounded away from zero or one so that the following inequality holds~($\forall i,j,k,\Ical$ with $i<j$ and $k\in\Jcal$).
\begin{align}
\frac{\alpha_i \sum_{\Ical\subset\Jcal,\Ical\ni k}\Pr(\Q{i}_2=\Ical)\P_{\bar{\Jcal}\cup\{k\},\Jcal \setminus \Ical}}{\alpha_j \Pr(\Q{j}_2=\Jcal) (1-\delta_{\bar{\Jcal}\cup\{k\}})} \leq \Pr(\Q{j}_1=i)  .
\end{align}
There are $j-1$ such inequalities for each given set of $(j,k,\Jcal)$. In addition, $\Pr(\Q{j}_1=i) $ should also fulfil $0\leq \Pr(\Q{j}_1=i) \leq 1$ and $\sum_{i=1}^{j-1} \Pr(\Q{j}_1=i)=1$. Then, we can eliminate the set $\{\Pr(\Q{j}_1=i)\}_{i=1\ldots j-1}$ with the Fourier–Motzkin elimination~(FME) to obtain $K-1$ constraints on the $\alpha$'s. Let us take $i=1$ as an example, as show below.
\begin{align}
\frac{\alpha_1 \sum_{\Ical\subset\Jcal,\Ical\ni k}\Pr(\Q{i}_2=\Ical)\P_{\bar{\Jcal}\cup\{k\},\Jcal \setminus \Ical}}{\alpha_j \Pr(\Q{j}_2=\Jcal) (1-\delta_{\bar{\Jcal}\cup\{k\}})} \leq & \Pr(\Q{j}_1=1),\\
0\leq & \Pr(\Q{j}_1=1),\\
& \Pr(\Q{j}_1=1) \leq 1-\sum_{i=2}^{j-1} \Pr(\Q{j}_1=i).
\end{align}
We obtain 
\begin{align}
\frac{\alpha_1 \sum_{\Ical\subset\Jcal,\Ical\ni k}\Pr(\Q{i}_2=\Ical)\P_{\bar{\Jcal}\cup\{k\},\Jcal \setminus \Ical}}{\alpha_j \Pr(\Q{j}_2=\Jcal) (1-\delta_{\bar{\Jcal}\cup\{k\}})} \leq  1-\sum_{i=2}^{j-1} \Pr(\Q{j}_1=i).
\end{align}
Therefore, we can apply $j-1$ times the same type FME and the rate constraints after these FME are
\begin{align}
  0&\leq \alpha_j \Pr(\Q{j}_2=\Jcal) (1-\delta_{\bar{\Jcal}\cup\{k\}}) -
  \sum_{i=1}^{j-1}\alpha_i\sum_{\Ical\subset\Jcal,\Ical\ni k}\Pr(\Q{i}_2=\Ical)\P_{\bar{\Jcal}\cup\{k\},\Jcal\setminus\Ical},  \label{ch-Kuser:eq:alphaj}
\end{align}
for $j=2,\ldots,K$. This completes the proof.


\section{Proof of Proposition~\ref{CH-KUSER:CORO:SYMSPIND}}\label{ch-Kuser:appendix:EBCsym}
We define 
$\mu_{\Jcal} \defeq \alpha_{j}P(\Q{j}_2=\Jcal)$ as the normalized length such that $\sum_{\Jcal:
\Jcal\subseteq\Kcal} \mu_{\Jcal} = 1$. The rate region in Corollary~\ref{CH-KUSER:COL:ERA} can be rewritten as
\begin{align}
R_{k} &\leq \mu_{\{k\}} (1-\delta_{\Kcal}) \log \vert \Xcal \vert,\label{appendB:eq:Rk}\\
\mu_{\Jcal}  &\geq \max_{k:\;k\in\Jcal} \left\{\sum_{\Ical:\;k\in\Ical\subset\Jcal}
\frac{\P_{\bar{\Jcal}\cup\{k\},\Jcal\setminus\Ical}}{1- \delta_{\bar{\Jcal}\cup\{k\}}}\mu_{\Ical}\right\}. \label{appendB:eq:alphaJ}
\end{align}
First, we show that \eqref{appendB:eq:alphaJ} should be satisfied with equality for all $\Jcal\subseteq \Kcal$.
It follows the similar steps as those in Appendix~\ref{ch-Kuser:appendix:Gaussianregion}. We assume that there
exist $\{\mu_{\Jcal}\}_{\Jcal\subseteq\Kcal}$ such that 
$\sum_{\Jcal\subseteq \Kcal}\mu_{\Jcal} = 1$ holds, and that the inequality 
\eqref{appendB:eq:alphaJ} is strict for some ${\Jcal^{'}}$. In this case, one can always reduce the value of
$\mu_{\Jcal^{'}}$ to achieve equality in \eqref{appendB:eq:alphaJ}, which leads to a smaller sum $\sum
\mu_{\Jcal} = c < 1$. Then, we can scale the whole set $\{\mu_{\Jcal}\}_{\Jcal\subseteq\Kcal}$ by $c$ to make
sure that $\sum \mu_{\Jcal} = 1$ holds again. This will increase the values of $\{\mu_{\{k\}}\}_{k\in\Kcal}$ by
a factor $\frac{1}{c}$, and will increase simultaneously the rate in \eqref{appendB:eq:Rk}. Therefore, it is
without loss of optimality to assume that  \eqref{appendB:eq:alphaJ} is satisfied with equality. 

Then, we focus on the symmetric EBC,  for which the optimal normalized lengths are characterized by the following lemma.
\begin{lemma}\label{ch-Kuser:lm:property}
  Let us define $k^*_{\Jcal} \defeq \min_{k\in\Jcal} k$ and 
  \begin{align}
    \mu_{k,\Jcal} \defeq \sum_{\Ical:\;k\in\Ical\subset\Jcal}
    \frac{\P_{\bar{\Jcal}\cup\{k\},\Jcal\setminus\Ical}}{1- \delta_{\bar{\Jcal}\cup\{k\}}}\mu_{\Ical}.
    \label{eq:tmp91}
  \end{align}%
For a symmetric EBC, the optimal $\mu_{\Jcal}$, $\Jcal\subseteq\Kcal$, is 
\begin{align}
\mu_{\Jcal} = \mu_{k^{*}_{\Jcal},\Jcal}. \label{ch-Kuser:eq:property1}
\end{align}
\end{lemma}
\begin{proof}
  We prove the lemma by induction on $j$. Note that by definition $\mu_{k^*_{\Ical},\Ical} =
  \mu_{k^*_{\Jcal},\Ical}$ for $\Ical\subset\Jcal$. We define two sets $\Jcal_1$ and $\Jcal_2$ that verify
$\vert\Jcal_1\vert=\vert\Jcal_2\vert=j=\vert\Jcal\vert \geq 2$ and
$\Jcal_1\setminus\{k^*_{\Jcal_1}\}=\Jcal_2\setminus\{k^*_{\Jcal_2}\}$.
If $k^*_{\Jcal_1} \leq k^*_{\Jcal_2}$, it
can be proved with induction that $\mu_{\Jcal_1} \geq \mu_{\Jcal_2}$. To initiate the induction, we assume
that the maximal $\mu_{\Ical}$ is obtained with $k^*_{\Ical}$ and $\mu_{\Ical_1} \geq \mu_{\Ical_2}$ is
correct with the analogously defined $\Ical_1,\Ical_2$. Due to the channel's symmetry, we assume without loss of
generality that 
$\mu_{\{1\}}\geq\mu_{\{2\}}\geq\cdots\geq\mu_{\{K\}}$ and we use abusively the following notations in this appendix
$\delta_{\bar{\Jcal}\cup\{k^*_{\Jcal}\}}=\delta_{K-j+1}$ and
$\P_{\bar{\Jcal}\cup\{k^*_{\Jcal}\},\Jcal\setminus\Ical}=\P_{K-j+1,j-i}$. Hence, \eqref{eq:tmp91} and
\eqref{ch-Kuser:eq:property1}
reduce to $\mu_{k^{*}_{\Jcal},\Jcal} = \frac{1}{1- \delta_{K-j+1}} \sum_{k^*_{\Jcal}
\in\Ical\subset\Jcal} \mu_{k^*_{\Ical},\Ical} \P_{K-j+1,j-i}$. 

As Lemma~\ref{ch-Kuser:lm:property} focuses on $j\ge 2$ case, we start by verify the case with $j=2$. We assume
that $\mu_{\Jcal_1}=\mu_{\{t_1,t_3\}}$ and $\mu_{\Jcal_{2}}=\mu_{\{t_2,t_3\}}$ where $t_1\leq t_2\leq t_3$. Then,
we notice that $\mu_{\Jcal_m}=\mu_{\{t_m,t_3\}} = \frac{\P_{K-1,1}}{1- \delta_{K-1}} \max_{k\in\{t_m,t_3\}}
\mu_{\{k\}} = \frac{\P_{K-1,1}}{1- \delta_{K-1}}  \mu_{\{t_m\}}$, $m=1,2$ and $\mu_{\Jcal_1}=\frac{\P_{K-1,1}}{1-
\delta_{K-1}}  \mu_{\{t_1\}}\geq \frac{\P_{K-1,1}}{1- \delta_{K-1}}  \mu_{\{t_2\}} =\mu_{\Jcal_2}$. 

Let us assume that \eqref{ch-Kuser:eq:property1} and $\mu_{\Jcal_1} \geq \mu_{\Jcal_2}$ hold for any
$\Jcal_m\subset \Kcal$ with $ \vert\Jcal_m\vert = j = l-1$~($3 \leq l \leq K$), $k^*_{\Jcal_1} \leq k^*_{\Jcal_2}$,
and $\Jcal_1\setminus\{k^*_{\Jcal_1}\}=\Jcal_2\setminus\{k^*_{\Jcal_2}\}$, for $m=1,2$. We show that
\eqref{ch-Kuser:eq:property1} and $\mu_{\Lcal_1} \geq \mu_{\Lcal_2}$ hold for any $\Jcal_m\subset \Lcal_m \subseteq
\Kcal$ with $\vert\Lcal_m\vert = l$, $k^*_{\Lcal_1} \leq k^*_{\Lcal_2}$, and $\Lcal_1
\setminus\{k^*_{\Lcal_1}\}=\Lcal_2\setminus\{k^*_{\Lcal_2}\}$, for $m=1,2$. Let us take $\mu_{\Lcal_1}$ as an
example. The $\mu_{\Lcal_1}$ writes as
\begin{align}
\mu_{\Lcal_1}  &= \frac{\P_{K-l+1,l-j}}{1- \delta_{K-l+1}}
\max_{k\in\Lcal_1} \sum_{k\in\Jcal_1\subset\Lcal_1} \mu_{\Jcal_1} \\
&= \frac{\P_{K-l+1,l-j}}{1- \delta_{K-l+1}} \max
\left\{\sum_{k^*_{\Lcal_1}\in\Jcal_1\subset\Lcal_1} \mu_{\Jcal_1} , \left\{ \sum_{k'\in\Jcal'\subset\Lcal_1} \mu_{\Jcal'} \right\}_{k'\neq k^*_{\Lcal_1}}\right\}\\
&= \frac{\P_{K-l+1,l-j}}{1- \delta_{K-l+1}} \max_{k'\neq k^*_{\Lcal_1},\atop k'\in\Lcal_1}\left\{\max
\left\{\sum_{k^*_{\Lcal_1}\in\Jcal_1\subset\Lcal_1} \mu_{\Jcal_1} , \sum_{k'\in\Jcal'\subset\Lcal_1} \mu_{\Jcal'}
\right\}\right\}\label{ch-Kuser:eq:syminduction1}\\
&= \frac{\P_{K-l+1,l-j}}{1- \delta_{K-l+1}} \sum_{k^*_{\Lcal_1}\in\Jcal_1\subset\Lcal_1}
\mu_{\Jcal_1}. \label{ch-Kuser:eq:syminduction2}
\end{align}
To prove \eqref{ch-Kuser:eq:syminduction2}, we consider four types of subsets of $\Lcal_1$
depending on whether $k^*_{\Lcal_1}$ and $k'$ are included in the subset. In particular, a subset including
both $k^*_{\Lcal_1}$ and $k'$ appears in both terms inside the inner maximization of
\eqref{ch-Kuser:eq:syminduction1} which yields $k^*_{\Jcal'}=k^*_{\Jcal_1}=k^*_{\Lcal_1}$, while a subset containing
neither $k^*_{\Lcal_1}$ nor $k'$ does not appear inside the inner maximization of \eqref{ch-Kuser:eq:syminduction1}.
Note that the other subsets have either $k^*_{\Lcal_1}$ or $k'$ such that $k'\not\in\Jcal_1$ and
$k^*_{\Lcal_1}\not\in\Jcal'$. There always exists a mapping that projects a subset $\Jcal_1$ including
$k^*_{\Lcal_1}$ into another subset $\Jcal'$ including $k'$ by substituting $k^*_{\Lcal_1}$ for $k'$, i.e.,
$\Jcal'=\Jcal_1 \setminus \{k^*_{\Lcal_1}\}\cup\{k'\}$. Thus, $\sum_{k'\in\Jcal'\subset\Lcal_1} \mu_{\Jcal'}\leq
\sum_{k^*_{\Lcal_1}\in\Jcal_1\subset\Lcal_1} \mu_{\Jcal_1} $ holds for any $k'\neq k^*_{\Lcal_1}$ given that
$k^*_{\Lcal_1}\leq k'$ and the property $\mu_{\Jcal_1} \geq \mu_{\Jcal_2}$ is true for $k^*_{\Jcal_1} \le k^*_{\Jcal_2}$. Therefore, \eqref{ch-Kuser:eq:syminduction2} holds. The proof completes by
\begin{align}
\sum_{\Ical\subset\Lcal_1\setminus\{k^*_{\Lcal_1}\}} \mu_{\Ical\cup\{k^*_{\Lcal_1}\}}  \geq
\sum_{\Ical\subset\Lcal_2\setminus\{k^*_{\Lcal_2}\}} \mu_{\Ical\cup\{k^*_{\Lcal_2}\}} \Longleftrightarrow
\mu_{\Lcal_1} \ge  \mu_{\Lcal_2},
\end{align}
where the inequality follows by identifying $\Jcal_1=\Ical\cup\{k^*_{\Lcal_1}\}$,
$\Jcal_2=\Ical\cup\{k^*_{\Lcal_2}\}$ and the property $\mu_{\Jcal_1} \geq \mu_{\Jcal_2}$ is true when $k^*_{\Jcal_1}=k^*_{\Lcal_1} \le k^*_{\Lcal_2}= k^*_{\Jcal_2}$ and $\Jcal_1\setminus\{k^*_{\Jcal_1}\}=\Ical=\Jcal_2\setminus\{k^*_{\Jcal_2}\}$.
\end{proof}

In the following, we show that the $\mu$'s in Lemma~\ref{ch-Kuser:lm:property} lead to the capacity region
\eqref{CH-KUSER:CORO:SYMSPIND}. From
Lemma 10 in \cite{gatzianas2013multiuser}, we know that, for any disjoint sets $\Fcal, \Tcal \subseteq \Kcal$,
\begin{align}
\P_{\Fcal,\Tcal}=\sum_{\Ucal\subseteq\Tcal} (-1)^{\vert
\Ucal \vert} \delta_{\Fcal \cup \Ucal} =\sum_{\Ucal\subseteq\Tcal} (-1)^{\vert \Ucal \vert+1} (1-\delta_{\Fcal \cup
\Ucal}).
\end{align}%
Then, we extend $\P_{\bar{\Jcal}\cup\{k\},\Jcal\setminus\Ical}$ analogously such that \eqref{ch-Kuser:eq:property1} writes as
\begin{align}
  \mu_{\Jcal} &= (1- \delta_{\bar{\Jcal}\cup \{k^{*}_{\Jcal}\}})^{-1}\sum_{\Ical:\;k^*_{\Jcal}
  \in\Ical\subset\Jcal} \mu_{\Ical} \sum_{\Ucal:\;\Ucal\subseteq\Jcal\setminus\Ical} (-1)^{\vert \Ucal \vert+1}
  (1-\delta_{\bar{\Jcal} \cup \{k^{*}_{\Jcal}\} \cup \Ucal}) \nonumber  \\
  &= (1- \delta_{\bar{\Jcal} \cup \{k^{*}_{\Jcal}\}})^{-1}\sum_{k^*_{\Jcal} \in\Ical\subset\Jcal}  \mu_{\Ical}
  \bigg( -(1-\delta_{\bar{\Jcal} \cup \{k^{*}_{\Jcal}\} })+ \sum_{\phi \neq \Ucal\subseteq\Jcal \setminus\Ical}
  (-1)^{\vert \Ucal \vert+1} (1-\delta_{\bar{\Jcal} \cup \{k^{*}_{\Jcal}\} \cup \Ucal})\bigg) \nonumber \\
&= \frac{\sum_{\phi \neq \Ucal\subseteq\Jcal \setminus\{k^*_{\Jcal}\}} \sum_{k^*_{\Jcal}
\in\Ical\subseteq\Jcal\setminus\Ucal} \mu_{\Ical}(-1)^{\vert \Ucal \vert+1} (1-\delta_{\bar{\Jcal} \cup
\{k^{*}_{\Jcal}\} \cup \Ucal})}{1- \delta_{\bar{\Jcal} \cup \{k^{*}_{\Jcal}\} }} -\sum_{k^*_{\Jcal}
\in\Ical\subset\Jcal} \mu_{\Ical},   \label{ch-Kuser:eq:alphakJrecursion}
\end{align}
where we change the summation order over $\Ical$ and $\Ucal$ to obtain \eqref{ch-Kuser:eq:alphakJrecursion}. We
simplify \eqref{ch-Kuser:eq:alphakJrecursion} by adding $\sum_{k^*_{\Jcal} \in\Ical\subset\Jcal} \mu_{\Ical}$ to both sides of \eqref{ch-Kuser:eq:alphakJrecursion}, as shown below. 
\begin{align}
\sum_{\Ical:\;k^*_{\Jcal} \in\Ical\subseteq\Jcal} \mu_{\Ical}  = \frac{\sum_{\Ucal:\;\phi \neq
\Ucal\subseteq\Jcal\setminus\{k^*_{\Jcal}\}} (-1)^{\vert \Ucal \vert+1} (1-\delta_{\bar{\Jcal}\cup\{k^*_{\Jcal}\}
\cup \Ucal}) \sum_{\Ical:\;k^*_{\Jcal} \in\Ical\subseteq\Jcal\setminus\Ucal} \mu_{\Ical}}{1-
\delta_{\bar{\Jcal}\cup\{k^*_{\Jcal}\}}}.   \label{ch-Kuser:eq:recursion2}
\end{align}

Next, we show that the $\mu$'s satisfying the recursive relation, i.e., \eqref{ch-Kuser:eq:recursion2}, also verify Lemma~\ref{ch-Kuser:lm:recursion}.
\begin{lemma}\label{ch-Kuser:lm:recursion}
For a given $k$ and for any $\mathcal{W}_{k}$ such that $k\in\mathcal{W}_k \subseteq \{k,k+1,\cdots,K\}$, we have 
\begin{align}
\sum_{\Ical:\;k\in\Ical\subseteq \mathcal{W}_k} \mu_{\Ical} =
\frac{(1-\delta_{\Kcal})\mu_{\{k\}}}{1-\delta_{\Kcal\setminus\mathcal{W}_k\cup\{k\}}}. \label{ch-Kuser:eq:summationrecursion}
\end{align}
\end{lemma}
\begin{proof}
The proof is done by induction on the cardinality of $\mathcal{W}_k$. For arbitrary $k$ and
$\vert\mathcal{W}_k\vert=1$~(i.e., $\mathcal{W}_k=\{k\}$), one can easily verify \eqref{ch-Kuser:eq:summationrecursion} is true. 
We now assume that \eqref{ch-Kuser:eq:summationrecursion} holds for all $\mathcal{W}_k$ with
$\vert\mathcal{W}_k\vert \leq w$ and show that it also holds for all $\mathcal{W}_k$ with $\vert \mathcal{W}_k
\vert=w+1$. Note that $\Jcal$ is a set whose minimal element is $k^*_{\Jcal}$ and $\mathcal{W}_k$ is a set with its
minimum being $k$. Since \eqref{ch-Kuser:eq:recursion2} is true for any $\Jcal$, we can substitute $k$ and
$\mathcal{W}_k$ for $k^*_{\Jcal}$ and $\Jcal$, respectively, in \eqref{ch-Kuser:eq:recursion2} and have
\begin{align}
\sum_{k \in \Ical\subseteq\mathcal{W}_k} \mu_{\Ical}  & = (1- \delta_{\Kcal \setminus \mathcal{W}_k \cup
\{k\}})^{-1} \sum_{\phi \neq \Ucal\subseteq \mathcal{W}_k\setminus\{k\}} (-1)^{\vert \Ucal \vert+1}
(1-\delta_{\Kcal\setminus\mathcal{W}_k\cup\{k\} \cup \Ucal}) \sum_{k \in\Ical\subseteq \mathcal{W}_k \setminus
\Ucal} \mu_{\Ical}  \nonumber \\
& =   (1- \delta_{\Kcal \setminus \mathcal{W}_k \cup \{k\}})^{-1} \sum_{\phi \neq \Ucal\subseteq
\mathcal{W}_k\setminus\{k\}} (-1)^{\vert \Ucal \vert+1} (1-\delta_{\Kcal\setminus\mathcal{W}_k\cup\{k\} \cup
\Ucal}) \frac{(1-\delta_{\Kcal})\mu_{\{k\}}}{(1-\delta_{\Kcal\setminus\mathcal{W}_k\cup\{k\} \cup \Ucal})} \nonumber\\
& =  \frac{(1-\delta_{\Kcal})\mu_{\{k\}}}{1- \delta_{\Kcal\setminus\mathcal{W}_k\cup\{k\}}}  \sum_{\phi \neq
\Ucal\subseteq \mathcal{W}_k \setminus\{k\}} (-1)^{\vert \Ucal \vert+1} \nonumber \\
&=  \frac{(1-\delta_{\Kcal})\mu_{\{k\}}}{1- \delta_{\Kcal\setminus\mathcal{W}_k\cup\{k\} \cup \Ucal}},\label{ch-Kuser:eq:recursumalphaJ}
\end{align}
where second equality holds because of the assumption and \eqref{ch-Kuser:eq:recursumalphaJ} follows the binomial theorem.
\end{proof}

We sum up the $\mu_{\Jcal}$'s~($\sum_{\Jcal\subseteq \Kcal}\mu_{\Jcal} = 1$) to obtain $1 = \sum_{\Jcal\subseteq
\Kcal}\mu_{\Jcal}  = \sum_{k=1}^{K}\sum_{k\in\Ical\subseteq \{k,k+1,\cdots,K\}} \mu_{\Ical} = \sum_{k=1}^{K}
\frac{(1-\delta_{\Kcal})\mu_{\{k\}}}{1- \delta_{\Kcal\setminus\{k,k+1,\cdots,K\}\cup\{k\}}}=\sum_{k=1}^{K}
\frac{(1-\delta_{\Kcal})\mu_{\{k\}}}{1- \delta_{\{1,2,\cdots,k\}}}$, where the third equality comes from
Lemma~\ref{ch-Kuser:lm:recursion}. Additionally, we rewrite \eqref{appendB:eq:Rk} as
$\frac{R_{k}}{(1-\delta_{\Kcal})\log \vert \Xcal \vert} \leq \mu_{\{k\}}$, and apply this inequality to the above
sum, which yields $\log \vert \Xcal \vert \geq \sum_{k=1}^{K} \frac{R_{k}}{1- \delta_{\{1,2,\cdots,k\}}}$. 
The above proof also holds if we swap the roles of the users according to the permutation $\pi$. This completes the
proof.

\bibliographystyle{IEEEtran}
\bibliography{references}

\begin{thebibliography}{10}
\providecommand{\url}[1]{#1}
\csname url@samestyle\endcsname
\providecommand{\newblock}{\relax}
\providecommand{\bibinfo}[2]{#2}
\providecommand{\BIBentrySTDinterwordspacing}{\spaceskip=0pt\relax}
\providecommand{\BIBentryALTinterwordstretchfactor}{4}
\providecommand{\BIBentryALTinterwordspacing}{\spaceskip=\fontdimen2\font plus
\BIBentryALTinterwordstretchfactor\fontdimen3\font minus
  \fontdimen4\font\relax}
\providecommand{\BIBforeignlanguage}[2]{{%
\expandafter\ifx\csname l@#1\endcsname\relax
\typeout{** WARNING: IEEEtran.bst: No hyphenation pattern has been}%
\typeout{** loaded for the language `#1'. Using the pattern for}%
\typeout{** the default language instead.}%
\else
\language=\csname l@#1\endcsname
\fi
#2}}
\providecommand{\BIBdecl}{\relax}
\BIBdecl

\bibitem{Caire-Jindal-Kobayashi-Ravindran-TIT10}
G.~Caire, N.~Jindal, M.~Kobayashi, and N.~Ravindran, ``{Multiuser {MIMO}
  achievable rates with downlink training and channel state feedback},''
  \emph{IEEE Transaction on Information Theory}, vol.~56, no.~6, pp.
  2845--2866, May. 2010.

\bibitem{wang2012capacity}
C.-C. Wang, ``On the capacity of 1-to-{K} broadcast packet erasure channels
  with channel output feedback,'' \emph{IEEE Transaction on Information
  Theory}, vol.~58, no.~2, pp. 931--956, Feb. 2012.

\bibitem{gatzianas2013multiuser}
M.~Gatzianas, L.~Georgiadis, and L.~Tassiulas, ``Multiuser broadcast erasure
  channel with feedback-capacity and algorithms,'' \emph{IEEE Transaction on
  Information Theory}, vol.~59, no.~9, pp. 5779--5804, May. 2013.

\bibitem{maddah2012completely}
M.~A. Maddah-Ali and D.~Tse, ``Completely stale transmitter channel state
  information is still very useful,'' \emph{IEEE Transaction on Information
  Theory}, vol.~58, no.~7, pp. 4418--4431, April. 2012.

\bibitem{yi2013precoding}
X.~Yi and D.~Gesbert, ``Precoding methods for the {MISO} broadcast channel with
  delayed {CSIT},'' \emph{IEEE Transactions on Wireless Communications},
  vol.~12, no.~5, pp. 1--11, May. 2013.

\bibitem{ali2014approximate}
M.~A. Maddah-Ali and A.~S. Avestimehr, ``Approximate capacity region of the
  {MISO} broadcast channels with delayed {CSIT},'' \emph{IEEE Transaction on
  Communications}, vol.~64, no.~7, pp. 2913 -- 2924, June. 2016.

\bibitem{wangjue2013precoder}
J.~Wang, M.~Matthaiou, S.~Jin, and X.~Gao, ``Precoder design for multiuser
  {MISO} systems exploiting statistical and outdated {CSIT},'' \emph{IEEE
  Transactions on Communications}, vol.~61, no.~11, pp. 4551--4564, Sept. 2013.

\bibitem{clerckx2015space}
B.~Clerckx and D.~Gesbert, ``Space-{T}ime encoded {MISO} broadcast channel with
  outdated {CSIT}: An error rate and diversity performance analysis,''
  \emph{IEEE Transactions on Communications}, vol.~63, no.~5, pp. 1661--1675,
  Mar. 2015.

\bibitem{he2014capacity}
C.~He, S.~Yang, and P.~Piantanida, ``On the capacity of the fading broadcast
  channel with state feedback,'' in \emph{IEEE International Symposium on
  Communications, Control and Signal Processing (ISCCSP)}, May. 2014.

\bibitem{shayevitz2013capacity}
O.~Shayevitz and M.~Wigger, ``On the capacity of the discrete memoryless
  broadcast channel with feedback,'' \emph{IEEE Transaction on Information
  Theory}, vol.~59, no.~3, pp. 1329--1345, Mar. 2013.

\bibitem{kim2015note}
H.~Kim, Y.-K. Chia, and A.~El~Gamal, ``A note on the broadcast channel with
  stale state information at the transmitter,'' \emph{IEEE Transaction on
  Information Theory}, vol.~61, no.~7, pp. 3622--3631, May. 2015.

\bibitem{venkataramanan2013achievable}
R.~Venkataramanan and S.~Pradhan, ``An achievable rate region for the broadcast
  channel with feedback,'' \emph{IEEE Transaction on Information Theory},
  vol.~59, no.~10, pp. 6175--6191, June. 2013.

\bibitem{wu2016coding}
Y.~Wu and M.~Wigger, ``Coding schemes with rate-limited feedback that improve
  over the no feedback capacity for a large class of broadcast channels,''
  \emph{IEEE Transactions on Information Theory}, vol.~62, no.~4, pp.
  2009--2033, Feb. 2016.

\bibitem{tuncel2006slepian}
E.~Tuncel, ``Slepian-{W}olf coding over broadcast channels,'' \emph{IEEE
  Transaction on Information Theory}, vol.~52, no.~4, pp. 1469--1482, April.
  2006.

\bibitem{cover2012elements}
T.~M. Cover and J.~A. Thomas, \emph{Elements of {I}nformation {T}heory}.\hskip
  1em plus 0.5em minus 0.4em\relax John Wiley \& Sons, 2012.

\bibitem{elgamal_kim}
A.~{El Gamal} and Y.-H. Kim, \emph{Network {I}nformation {T}heory}.\hskip 1em
  plus 0.5em minus 0.4em\relax Cambridge University Press, 2011.

\bibitem{vaze2012dofnoCSIT}
C.~S. Vaze and M.~K. Varanasi, ``The {D}egree-of-{F}reedom regions of {MIMO}
  broadcast, interference, and cognitive radio channels with no {CSIT},''
  \emph{IEEE Transactions on Information Theory}, vol.~58, no.~8, pp.
  5354--5374, May. 2012.

\end{thebibliography}

\end{document}